\documentclass{article}
\usepackage[english]{babel}
\usepackage[english]{babel}
\usepackage{amsthm}
\usepackage{mathtools}% Loads amsmath
\newtheorem{theorem}{Theorem}[section]
\newtheorem{lemma}[theorem]{Lemma}

\usepackage{amsmath}
\usepackage{graphicx}
\usepackage[colorlinks=true, allcolors=blue]{hyperref}
\usepackage{cleveref}
\usepackage{booktabs}
\usepackage{siunitx}
\usepackage{natbib}
\usepackage{listings}
\usepackage{subfig}
\usepackage{bm}
\usepackage{pdfpages}
\usepackage{import}
\usepackage{float}
\usepackage{lipsum}
\usepackage{caption}
\usepackage{amssymb}
\usepackage{amsfonts}
\usepackage{blindtext}
\usepackage{titlesec}
\usepackage[doublespacing]{setspace}
\usepackage[fontsize=12pt]{fontsize}
\usepackage[margin=1in]{geometry}
\usepackage{mathrsfs}
\usepackage{bbm}
\usepackage{multirow}
\usepackage{graphicx} % Required for inserting images

\title{\textbf{Covariate-localized False Discovery Rates}}
\author{
  Jonathan Lin \\
  Department of Statistical Science, Duke University, Durham, NC \\
  jonathan.lin@duke.edu
  \and
  Surya Tokdar \\
  Department of Statistical Science, Duke University, Durham, NC \\
  surya.tokdar@duke.edu
}
\date{\today}

\begin{document}
\def\spacingset#1{\renewcommand{\baselinestretch}%
{#1}\small\normalsize} \spacingset{1}
\maketitle

\begin{abstract}
    We introduce a flexible model for covariate-dependent multiple testing which can be encoded using a nonparametric Gaussian mixture model. Weight-localized predictive recursion (PRx), a new development in the methodology of Newton's predictive recursion algorithm, is then leveraged to estimate the components of this mixture model, allowing for recovery of the covariate-localized false discovery rate $\text{Pr}(H_i = 0|z_i,x_i)$ using a single, unified algorithm. This quantity represents the most direct extension of Efron's local false discovery rate to the covariate-dependent setting, and admits provable Bayesian FDR control properties under simple rejection rules. We introduce several procedures for estimating and thresholding the local false discovery rate, and show using various simulations and a real-data example that our procedures lead to increased power, tighter Bayesian FDR control, and more interpretable rejections. We furthermore show that this holds for fixed and randomized hypothesis labels, indicating that our proposed methods perform well under both frequentist and Bayesian interpretations of multiple testing.
\end{abstract}
{{\textbf{\textit{Keywords---}} multiple testing, two groups model, predictive recursion}
\spacingset{1.72}
\newpage
\section{Introduction} \label{sec:intro}
%The principal goal of most multiple testing procedures is to control the false discovery rate, which is the expected proportion of false discoveries among all rejections \citep{benjamini1995controlling}. Towards this end, simple omnibus procedures such as the Benjamini-Hochberg rule admit provable global FDR control properties. 
Ever since the seminal work of \cite{benjamini1995controlling}, statistical literature on multiple testing has focused on controlling the proportion of false discoveries among all rejections. The original Benjamini-Hochberg rule and its many variants admit provable control on the false discovery rate, which is the expected proportion of false discoveries across repeated sampling. An alternative approach is given by thresholding Efron's \textit{local false discovery rates} \citep{efron2005local}, $l_i = \text{Pr}(H_i = 0 |z_i)$. This latter quantity underscores a Bayesian interpretation of multiple testing in which $z$-values are modeled as i.i.d. from the two-groups mixture, $m(z) = \pi_0f_0(z) +(1-\pi_0)f_1(z)$, with the indicator $H_i\sim\text{Bernoulli}(1-\pi_0)$ denoting the signal label. Under this model, the expected false discovery proportion is referred to as the Bayesian FDR (BFDR)---a quantity distinct from frequentist FDR, which treats the hypothesis labels as fixed. 

The presence of covariate information, in which test statistics $z_i$ are accompanied by covariates $x_i$, may lead to a sharp degradation in performance for procedures agnostic to the covariate \citep{scott2015false, chao2021adapt, ignatiadis2021covariate, leung2022zap, zhang2022covariate}. While procedures such as Benjamini-Hochberg maintain error rate control, their use of a static, fixed-$z$ rejection threshold can lead to a substantial loss in power. In this setting the \textit{covariate-localized false discovery rate}, 
$\ell_i=\ell(z_i,x_i):=\text{Pr}(H_i = 0|z_i,x_i)$,
offers a more comprehensive quantification of the post-data plausibility that $H_i$ is non-null. Several existing approaches exploit covariate information closely related to the covariate-localized false discovery rate to design more powerful FDR-controlling rules \citep{scott2015false, leung2022zap, chao2021adapt}; relatively little work, however, focuses on consistent recovery of $\ell(z,x)$. 

%,  articles among the covariate-dependent multiple testing \citep{scott2015false, chao2021adapt, ignatiadis2021covariate, leung2022zap, zhang2022covariate}---with many works offering provable FDR control---few treat $\ell_i$ as the target quantity, and to the best of our knowledge none offer provable asymptotic recovery of it. 

We propose a covariate-dependent extension of Efron's two-groups model that allows both the null proportion $\pi_0(x)$ and alternative density $f_1(z|x)$ to vary with the covariate $x$. We model $f_1(z\mid x)$ using a nonparametric Gaussian location mixture so that the entire conditional two-groups model can be represented by a covariate-dependent mixing distribution with a point mass at zero capturing the null effect and an absolutely continuous component capturing the non-null effects. Weight-localized predictive recursion (PRx)---a kernelized extension of Newton's \textit{predictive recursion} algorithm \citep{newton2002nonparametric, lin2026fastsemiparametricdensityregression}---is then used to estimate this mixing distribution, alongside the implied conditional density $m(z|x)$ of the test statistics. In particular, a single PRx fit simultaneously yields estimates of ${\pi}_0(x)$ and ${m}(z|x)$, leading to a direct estimate of the entire covariate-localized false discovery rate surface, $\ell(z,x)=\frac{\pi_0(x)f_0(z)}{m(z\mid x)}$. 
Treating $\ell(z,x)$ itself as the inferential target provides a useful point of comparison with modern covariate-adaptive methods, which exploit closely related information but differ fundamentally in the role assigned to the fitted model.

PRx shares important features with several recent covariate-adaptive testing methods; among these, ZAP \citep{leung2022zap} and AdaPT-GMM \citep{chao2021adapt} are particularly close comparators. Like PRx, these procedures exploit interactions between the covariates and the direction or shape of the test-statistic distribution. However, the role assigned to the fitted model in achieving FDR control is fundamentally different. Both ZAP and AdaPT-GMM use a fitted working model to guide an adaptive testing rule; frequentist FDR control is protected by a separate calibration or masking mechanism and therefore does not require correct specification of that model. PRx instead treats recovery of the underlying conditional mixing density as a primary inferential goal; our BFDR control guarantees derive from the consistency of its estimate. The numerical comparisons below examine these two routes to covariate-adaptive testing: FDR calibration that is insulated from the fitted working model, versus direct recovery of the conditional mixture on which the testing rule is based.

With estimates of $\ell_i=\ell(z_i,x_i)$ in hand, we study two simple rejection procedures. The first applies the oracle cumulative local false discovery rate thresholding rule from \cite{sun2007oracle} directly to the estimated $\ell_i$'s. At the oracle level, this rule  is power-optimal among rejection sets satisfying a posterior FDR constraint. The second uses two-fold cross-fitting to separate estimation of $\ell$ from its evaluation at the tested observations. We show that, under consistency of the PRx estimator and a mild condition ensuring a nonvanishing number of discoveries, the cross-fitted procedure achieves asymptotic BFDR control. 

%We also give a pathwise interpretation of the oracle rule: under the random-effects model, posterior FDP control translates almost surely into asymptotic control of the realized FDP. This latter result is distinct from uniform frequentist FDR control over arbitrary fixed configurations of null and non-null hypotheses.

Simulation studies are used to compare PRx with several modern covariate-adaptive multiple testing procedures. Across settings in which the prevalence, direction, or distribution of the signals varies with the covariates, PRx is often substantially more powerful while maintaining a mean false discovery proportion close to the nominal level. We additionally consider simulations with fixed hypothesis labels, which deliberately fall outside the random-effects assumptions used in our theory and therefore provide an empirical assessment of robustness to that form of model departure. Finally, we revisit the neural synchrony application of \cite{scott2015false}, where using PRx leads to a substantially different and scientifically interpretable set of discoveries.

\section{Methodology}\label{sec:methods}
Inspired by \cite{martin2012nonparametric}, we adopt a variant of Efron's two-groups model taking the form of a covariate-dependent spike-and-slab mixture. Let $x_i\sim \Pi$ i.i.d. and define $\pi_0(x):\mathcal{X}\rightarrow[0,1]$ to be a function that maps covariate values into the unit interval. Let the signal indicator be denoted by $H_i\sim \text{Bernoulli}(1-\pi_0(x_i))$, and let $z_i|x_i,H_i=0 \sim N(z;\theta_0, \sigma_0^2)$, $z_i|x_i, H_i=1\sim f_1(z|x) = \int N(z;\theta_0 + u, \sigma_0^2)\psi(u|x)\mu(du)$ for a nonparametric mixing density $\psi(u|x)$ supported on $\mathcal{U}$. This location-mixture formulation encodes Efron's \textit{zero assumption} by forcing the alternative density to have heavier tails than the null. We can write the density of $z|x$ as: $m(z|x)= \int N(z;\theta_0 + u, \sigma_0^2)\Psi(u|x)\mu(du)$, where $\Psi(u|x) = \pi_0(x)\delta_{0} + (1-\pi_0(x))\psi(u|x)$, $u\in[-C,C]$. From this we may express the \textit{covariate-localized false discovery rate} $\ell_i$ as:
\begin{equation}
    \begin{split}
        \ell_i = \text{Pr}(H_i = 0 |z_i,x_i) = \frac{\pi_0(x_i)N(z_i;\theta_0,\sigma_0^2)}{m(z_i|x_i)} = \frac{\Psi(\{0\}|x_i) N(z_i;\theta_0,\sigma_0^2)}{\int_\mathcal{U} N(z_i|\theta_0+u,\sigma_0^2) \Psi(u|x_i)\mu(du)}
    \end{split}
\end{equation}
Arguably, $\ell_i$ presents the most direct extension of Efron's local false discovery rate to the covariate-dependent setting. The goal of the following methodology is twofold. Firstly, given a stream of data $(z_i,x_i)$, we must determine a principled means of estimating $\ell_i$, which calls for an estimate of both $\pi_0(x)$ and $m(z|x)$. Secondly, after using these quantities to estimate $\ell_i$, we must determine a rejection rule using $\{\hat{\ell}_i\}_{i=1}^n$ that somehow controls the BFDR.

\subsection{Estimation of the covariate-localized false discovery rate}
\label{sec:lfdr-est}
%The first goal is handled by the marginal characterization of the covariate-dependent two-groups model. 
Since $N(z;\theta_0+u,\sigma_0^2)$ is a known parametric kernel, recovery of $\ell_i$ reduces to estimating the nonparametric mixing density, $\Psi$, which is dominated by the measure that assigns a point mass to $\{0\}$ and Lebesgue measure to the compact set $[-C,C]$. This can be pursued using PRx, which is designed for fast nonparametric recovery of a conditional mixing density. The algorithm is as follows: initialize a mixing density as $\hat{\Psi}_0(u|x) = \hat{\pi}_{0,0}(x)\delta_0+(1-\hat{\pi}_{0,0}(x))\hat{\psi}_0(u|x)$. For fixed $b_j$ and a stream of covariates $\{x_i\}_{i=1}^n$, $x_i = (x_{i,1},\ldots,x_{i,p})\in \mathbb{R}^p$, let $\beta_i(x^*) = \exp(-\sum_{j=1}^p b_j(x_{i,j} - x^*_j)^2 )$ and $v_i(x)=\beta_i(x) h(S_i(x))$, where $S_i(x) = \sum_{j=1}^i \beta_j(x)$ and $h(z)\asymp z^{-\gamma}$ for $\gamma\in(1/2,1]$. Then we iterate forwards according to the recursion:
\begin{equation}
    \begin{split}
        \hat{\Psi}_n(u|x) = (1-v_n(x))\hat{\Psi}_{n-1}(u|x)+ v_n(x)\frac{N(z_n;\theta_0+u,\sigma_0^2) \hat{\Psi}_{n-1}(u|x)}{\int_{\mathcal{U}} N(z_n;\theta_0+u,\sigma_0^2)\hat{\Psi}_{n-1}(u|x)\mu(du)}
    \end{split}
\end{equation}
Under mild regularity assumptions (see Section \ref{sec:theory}) this recursion leads to weak consistency $\hat{\Psi}_n\xrightarrow{w}\Psi$ regardless of the choice of dominating measure $\mu$. To eliminate order-dependency, one can repeat this process over several re-orderings of the data and then calculate a permutation-averaged estimate, which preserves consistency. After obtaining an estimate $\hat{\Psi}_{n}$ of the mixing density, we can easily extract $\hat{m}_{n}(z|x) = \int N(z;\theta_0 + u, \sigma_0^2)\hat{\Psi}_{n}(u|x)\mu(du)$ via quadrature. Estimation of the local false discovery rate $\text{Pr}(H_i = 0|z_i,x_i)$ can then be pursued in a variety of ways. We outline some options below.

\begin{itemize}
    \item (Naive PRx) Define the local false discovery rate estimator: $\hat{\ell}_n(z,x) = \frac{\hat{\Psi}_{n}(\{0\}|x) f_0(z)}{\hat{m}_{n}(z|x)}$ 
    \item (Hole-adjusted PRx) For some small $h>0$, let $\hat{\psi}_0(u|x)=0$ for $u\in[-h,h]$ and $\hat{\psi}_0(u|x)>0$ everywhere else. Let $\hat{\ell}_n(z,x) = \frac{\hat{\Psi}_{n}(\{0\}|x)f_0(z)}{\hat{m}_{n}(z|x)}$.
    \item (Conservative PRx) For any small $\eta> 0$ let $\hat{\ell}_n(z,x) = \min\left\{1,\frac{\hat{\Psi}_{n}([-\eta,\eta]|x)f_0(z)}{\hat{m}_{n}(z|x)}\right\}$. 
\end{itemize}

In Section \ref{sec:theory} we show that, under the assumption that the true continuous component satisfies $\psi(u|x) = 0 $ for $u\in[-\delta,\delta]$ for some $\delta\geq h$, hole-adjusted PRx is $L^p$-consistent towards the oracle target, $\ell(z,x)$. Conservative PRx asymptotically overestimates the oracle local false discovery rate, and the accrued error decreases as $\eta$ decreases; no separation assumption is necessary for this. Moreover, under separation of $\delta > \eta$, conservative PRx is exactly consistent. Naive PRx is the most straightforward between the three, is closely approximated by the other two, and performs very well empirically, but as of yet lacks the theoretical backing to suggest good asymptotic behavior. For additional details of the weight-localized predictive recursion algorithm, see \cite{lin2026fastsemiparametricdensityregression}. 

%To be concrete, let us define $\Psi_0$ as an initial ``guess.'' Weight-localized predictive recursion (PRx) then proceeds by updating according to the recursion, $m_{i-1}(z|x) = \int N(\theta_0+u,\sigma_0^2)\Psi_{i-1}(u|x)\mu(du)$, followed by $\Psi_i (z|x)= (1-v_i(x))\Psi_{i-1}(z|x) + v_i(x) \cdot \phi(z;\theta_0+u, \sigma_0^2)\Psi_{i-1}(z|x)/m_{i-1}(z|x)$. The weights $v_i(x)$ are defined as $v_i(x) = \beta_i(x)h(S_i(x))$, where $\beta_i(x) = \exp(-b\lVert  x-x_i\rVert^2)$ and $S_i(x) = \sum_{t=1}^i \beta_t(x)$. $h(z)$ is any function satisfying $h(z) \asymp z^{-\gamma}$ for $\gamma\in (1/2,1]$. We take $h(z) = (1+z)^{-2/3}$. The PRx weights $v_i(x)$ essentially encode the ``influence'' of a data point with covariate $x_i$ to the target conditional density at evaluation point $x$, compounded with the influence of previous data points that are decaying at the order of $z^{-\gamma}$. After running this recursion, our estimate is given simply by $\Psi_n(u|x)$. Following this, $\hat{\pi}_0(x) = \Psi_n(0|x)$, and $\hat{f}(z|x)$ can be computed using off-the-shelf numerical integration methods, such as quadrature. $\ell_i$ may then be estimated simply by plugging in these quantities.

There are three principal advantages of using PRx to estimate the local false discovery rate. Firstly, PRx offers a single, self-contained procedure for jointly estimating all components of $\ell_i$; recovery of $\hat{\Psi}_n(\cdot|x)$ implies recovery of $\hat{\pi}_{0,n}(x)$ and $\hat{m}_n(z|x)$ simultaneously, with no need for additional regression or EM steps. Secondly, PRx allows us to encode $f_1(z|x)$ as a nonparametric conditional Gaussian mixture, yielding a rich class of smooth alternative densities; it is well-known that any absolutely continuous density can be approximated by a continuous Gaussian mixture. PRx also admits as a byproduct a means of estimating a likelihood score function --- the PRMLx, $\prod_{i=1}^n \hat{m}_{i-1}(z_i|x_i)$ --- which may be maximized both for tuning its localization bandwidth, and for estimating unmixed parameters \citep{martin2012nonparametric, lin2026fastsemiparametricdensityregression}. For example, an optimal value for the localization bandwidth can be obtained by solving $\hat{b} = \text{arg max}_b \prod_{i=1}^n \hat{m}_{i-1}(z_i|x_i)$, and using PRMLx maximization to estimate optimal values for $\theta_0$ and $\sigma_0$ amounts to recovery of an \textit{empirical null} distribution. The empirical null is commonly used to adjust for violations of certain testing assumptions \citep{efron2004large}, or for the presence of an inflated central bulk in the data \citep{xiang2024interpretation}. See Section \ref{sec:neural} for an example.

These benefits distinguish PRx from existing covariate-adaptive multiple testing procedures. \cite{scott2015false} attempts to recover the $\ell(z,x)$ surface by first using ordinary predictive recursion to recover the alternative density of the $z$-values, and then plugging it into a separate EM algorithm for estimation of $\pi_0(x)$. PRx instead estimates both components jointly and permits the alternative distribution to vary with $x$. Procedures based on two-sided p-values such as IHW-BH and CAMT \citep{ignatiadis2021covariate,zhang2022covariate} may discard useful information when there is high correlation between $x$ and $\mathrm{sign}(z)$. ZAP and AdaPT-GMM address this issue more directly. ZAP \citep{leung2022zap} transforms the signed $z$-values and constructs a covariate-dependent assessor designed to mimic a useful local-FDR ordering, while AdaPT-GMM \citep{chao2021adapt} fits the $z$-values using a finite Gaussian mixture model with covariate-dependent weights. In both cases, however, the fitted model is primarily an efficiency device: their FDR guarantees are protected by a separate masking or calibration argument and do not require the working model to estimate the true local FDR consistently. PRx instead places greater inferential weight on recovering $\ell(z,x)$ itself.

%The work of \cite{chao2021adapt} and \cite{leung2022zap} remedy this in different ways; the latter method, for example, uses the transform $u_i = \Phi(z_i)$, and then models these $u$-values using a beta mixture. However, the alternative density of these $u$-values is modeled restrictively, with a static, user-specified shape parameter. \cite{chao2021adapt} fit the $z$-values directly to a discrete, $K$-component Gaussian mixture, $\sum_{k=1}^K w_k(x) N(z;\mu_k,\sigma_k^2)$---while certainly flexible, under the two-groups interpretation of multiple testing this leads to a loss of identifiability in the model parameters $\mu_k$ and $\sigma_k$ \citep{martin2012nonparametric}.

\subsection{Thresholding rules}
\label{sec:threshold}
After choosing a means with which to estimate the covariate-localized false discovery rates, we require a rejection rule with desirable error rate control properties. Let $\hat{\ell}_i:=\hat{\ell}_n(z_i,x_i)$. We use ``posterior FDR'' for the conditional quantity $\mathbb{E}[\mathrm{FDP}\mid\mathcal D_n]$, and ``BFDR'' for its repeated-sampling expectation under our covariate-dependent two-groups model. We propose two rejection rules which are outlined below:
\begin{itemize}
    \item (SC Thresholding Rule) Calculate $\{\hat{\ell}_{i}\}_{i=1}^n$ using the chosen PRx method and order them as $\hat{\ell}_{(1)}\leq\cdots\leq\hat{\ell}_{(n)}$, using a fixed tie-breaking rule if necessary. First calculate $k^*:=\max\{k:\frac{1}{k}\sum_{i=1}^k \ell_{(i)}\leq \alpha\}$, and then reject all hypotheses for which $\ell_i\leq \ell_{(k^*)}$.
    \item (Two-fold Thresholding Rule) Partition the indices $\{1,\ldots,n\}$ into two disjoint folds $I_1$ and $I_2$, both of which are fixed or chosen independently of the data with $\lim_{n\rightarrow\infty}|I_1|/n>0$ and $\lim_{n\rightarrow\infty}|I_2|/n>0$. Run the chosen PRx method on $(z_i,x_i)_{i\in I_1}$ to obtain $\hat{\ell}^{(1)}(z,x)$, and on $(z_i,x_i)_{i\in I_2}$ to obtain $\hat{\ell}^{(2)}(z,x)$. Then, calculate $\hat{\ell}^{(1)}(z_i,x_i);i\in I_2$ and $\hat{\ell}^{(2)}(z_i,x_i);i\in I_1$, so we have $n$ total local false discovery rate estimates. Apply the SC thresholding rule to this collection. 
\end{itemize}

The SC thresholding rule is simpler, does not rely on randomness beyond that which is present in the data, and approximates the rule that achieves power optimality among procedures that threshold the posterior FDR. The two-fold thresholding rule, on the other hand, can be shown to asymptotically control for the BFDR without relying on knowledge of an oracle target. As will be seen in Section \ref{sec:sims}, both rules exhibit similar performance across simulation studies. 

Both thresholding rules proposed above follow from the power-optimal oracle procedure of \cite{sun2007oracle}. The SC thresholding rule is identical to the oracle rule, except it replaces $\ell_i$ with the estimated quantity, $\hat{\ell}_i$. The two-fold thresholding rule employs a more careful data-splitting procedure to resolve the fact that the data pair $(z_i,x_i)$ appears twice in $\hat{\ell}_n(z_i,x_i)$: as both an evaluation point, and within the PRx recursion internal to $\hat{\ell}_n(\cdot,\cdot)$. In order to achieve provable BFDR control it becomes necessary to separate the data used in the PRx recursion from the evaluation point of the local false discovery rate estimate; data-splitting is one simple means of achieving this.  

\section{Case Studies}\label{sec:casestudies}
Here we present applications of PRx multiple testing to several simulations and a real-data example. Our goal in presenting these is to emphasize the benefit of learning the $\mathcal{X}$-dependency in the model. In particular we will show that PRx can capture shifts in the frequency of signals, as well as in their magnitude and sign, as they vary over the covariate space. Throughout the simulations we will use two variants of PRx:
\begin{itemize}
    \item (HA-PRx) Hole-adjusted PRx $(h = 0.01)$, with $\hat{\pi}_{0,0}(x) = 0.6$ and $\hat{\psi}_0(u|x)$ initialized as a uniform over $[-10, -h]\cup[h,10]$. The SC thresholding rule is used. We use PRMLx maximization to obtain an optimal localization bandwidth, $\hat{b}$.
    \item (Two-fold PRx) Hole-adjusted PRx $(h = 0.01)$, with $\hat{\pi}_{0,0}(x) = 0.6$ and $\hat{\psi}_0(u|x)$ initialized as a uniform over $[-10, -h]\cup[h,10]$. The two-fold thresholding rule is used.  PRMLx maximization is run twice---once on each fold---to obtain optimal localization bandwidths $\hat{b}$ for each fold. 
\end{itemize}
For both HA-PRx and two-fold PRx, the L-BFGS-B algorithm is used to carry out PRMLx maximization. These are then compared with predictive recursion expectation maximization (PREM) \citep{scott2015false}, adaptive thresholding (AdaPT-GMM) \citep{chao2021adapt}, cross-weighted multiple testing (IHW-BH) \citep{ignatiadis2021covariate}, $z$-value adaptive thresholding (ZAP) \citep{leung2022zap}, covariate-adaptive multiple testing (CAMT) \citep{zhang2022covariate}, Efron's simple two-groups procedure (Efron 2-grp) \citep{efron2005local}, and the naive Benjamini-Hochberg procedure (BH) \citep{benjamini1995controlling}. For the simulations we additionally apply the oracle rule, which calculates the true value of $\text{Pr}(H_i = 0|z_i,x_i)$ for each $(z_i,x_i)$ pair, and then applies the SC rule to them. It can be shown that naive, hole-adjusted, and conservative PRx lead to identical results for small values of $h$ and $\eta$; for this reason we forego using naive and conservative PRx in the following simulations, but include them in the supplementary.

\subsection{Simulation Studies}\label{sec:sims}
We consider three simulation studies, all of which use a simple covariate $x_i\sim \text{Uniform}[0,1]$ and the theoretical null, $z_i|x_i,H_i=0 \sim N(0,1)$. For the first simulation study, let $\pi_0(x) = 0.8 - 0.5\sin(\pi x_i)^2$, and let $z_i|x_i, H_i = 1 \sim N(-3,2)\mathbbm{1}\{x_i\leq 0.5\} + N(3,2)\mathbbm{1}\{x_i>0.5\}$. For the second simulation, let $\pi_0(x)=(1 + e^{4x+2})^{-1}$, with $z_i|x_i,H_i=1\sim N(\theta_i, 1)$, $\theta_i|x_i\sim N(\mu(x_i), 2)$ for $\mu(x) = -3+3x$ for $x\in[0,1/2]$ and $\mu(x) = 3x$ for $x\in(1/2,1]$. For the third simulation study, we let $\pi_0(x) = (1+e^{2-4x})^{-1}$ and let $z_i|x_i,H_i=1\sim N(\theta_i,1)$ with $\theta_i|x_i\sim N(\mu(x_i),1)$, with $\mu(x) = 3.5-x$.  We run our methods at $\alpha=0.1$ over $30$ replicates of the data, each with $n=1000$ hypotheses, and calculate the mean of the false discovery proportion and empirical power. The results are displayed in Table \ref{tab:simulation-bayesian}, alongside visualizations of the simulated data in Figure \ref{fig:testing}

\begin{figure}[ht]
\centering
\begin{minipage}{\textwidth}
  \centering
  \includegraphics[width=0.9\textwidth]{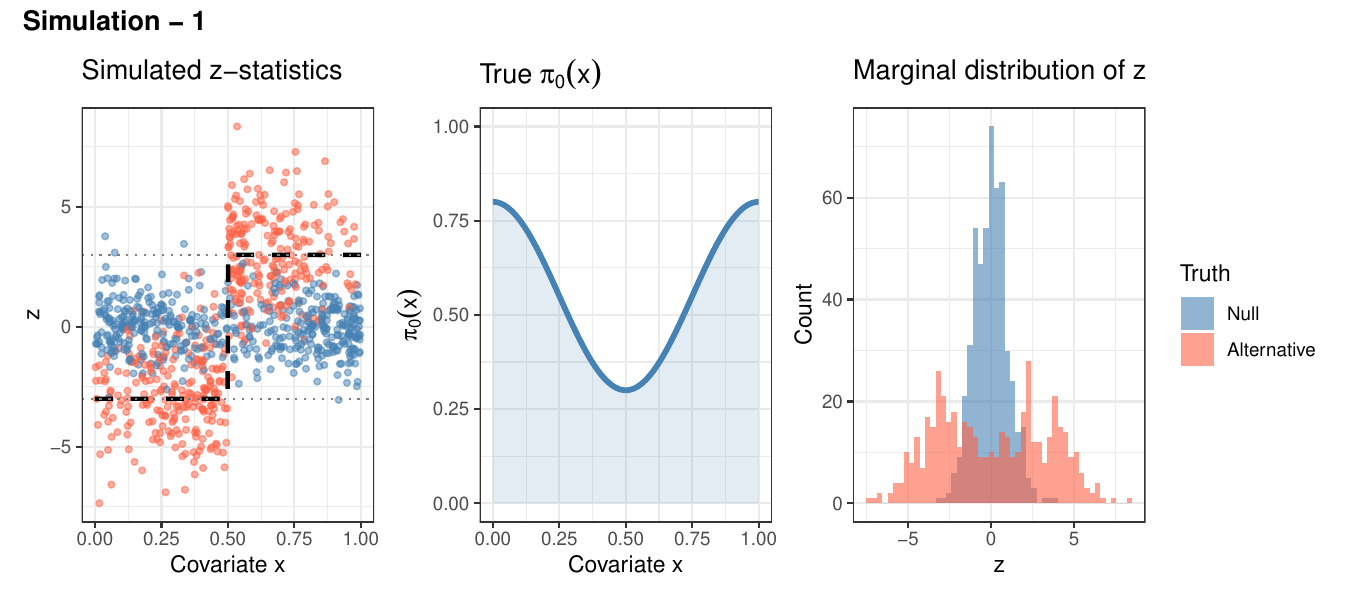}\\
\end{minipage}
\begin{minipage}{\textwidth}
  \centering
  \includegraphics[width=0.9\textwidth]{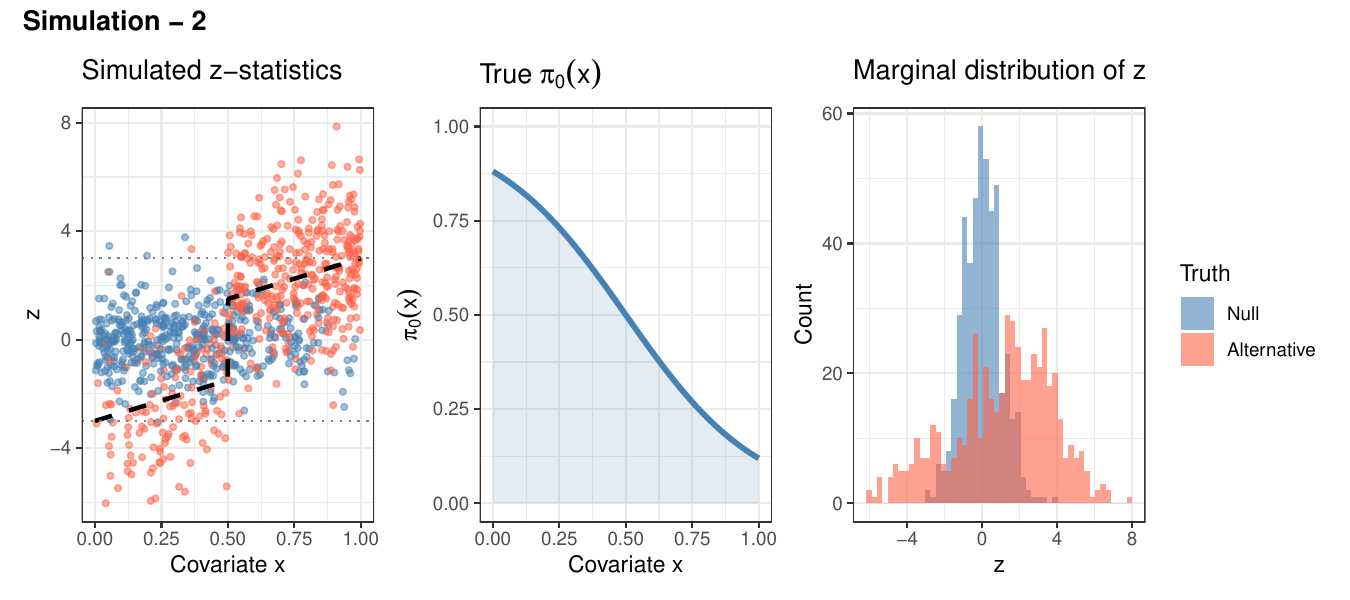}\\
\end{minipage}
\begin{minipage}{\textwidth}
  \centering
  \includegraphics[width=0.9\textwidth]{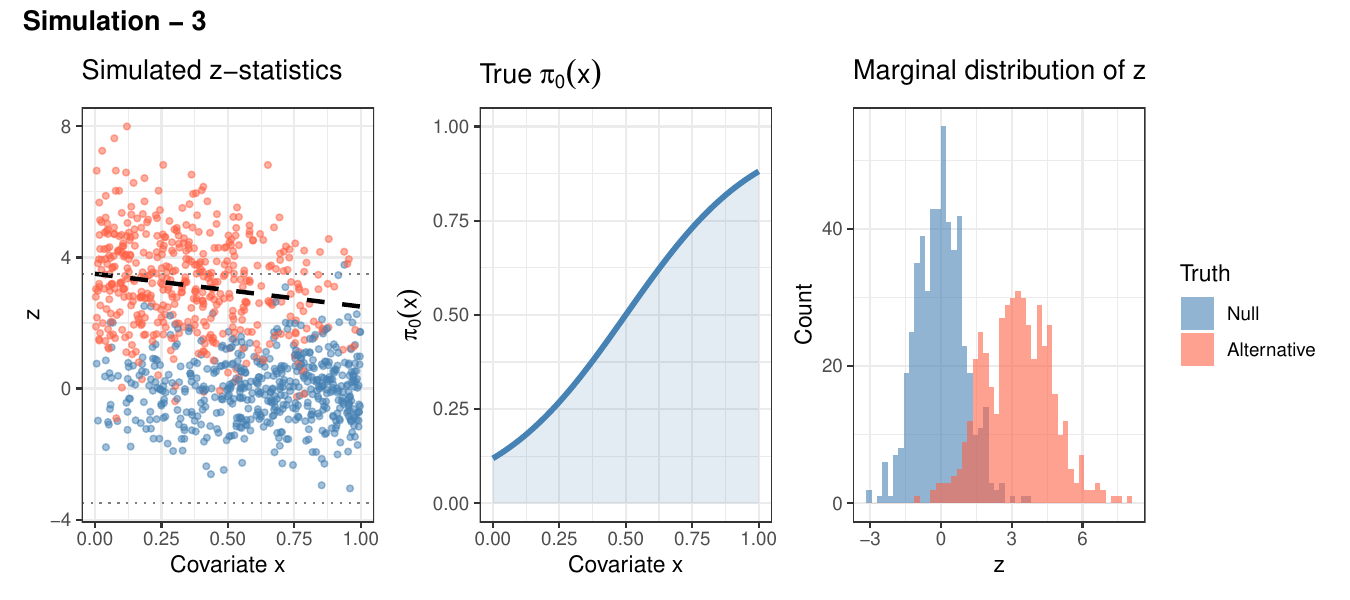}\\
\end{minipage}
\caption{Each row displays a simulation setting. For each row, the left panel displays a scatterplot of the observed data, $(x_i,z_i)_{i=1}^{1000}$ with the dashed line denoting the signal mean as it varies over the covariate space. The middle panel displays the curve $\pi_0(x)$, and the right panel displays a histogram of the observed $z$-values. }
\label{fig:testing}
\end{figure}

\begin{table}[ht]
\centering

\begin{tabular}{lccc}
\hline
Simulation 1 Method & Mean FDP (SD) & Mean Power (SD) & $n_{\mathrm{rejected}}$ \\
\hline
HA-PRx          & 0.10 (0.01) & 0.83 (0.02) & 412.1 \\
Two-fold PRx    & 0.10 (0.01) & 0.82 (0.02) & 405.6 \\
PREM            & 0.11 (0.02) & 0.78 (0.03) & 394.1 \\
AdaPT-GMM           & 0.09 (0.02) & 0.69 (0.03) & 337.9 \\
ZAP             & 0.09 (0.02) & 0.76 (0.03) & 374.9 \\
CAMT            & 0.09 (0.02) & 0.74 (0.03) & 366.9 \\
IHW-BH          & 0.07 (0.02) & 0.71 (0.03) & 346.0 \\
Efron-2grp      & 0.10 (0.02) & 0.75 (0.03) & 375.7 \\
BH              & 0.06 (0.01) & 0.70 (0.03) & 331.1 \\
Oracle          & 0.10 (0.01) & 0.84 (0.02) & 418.5 \\
\hline
\end{tabular}
\begin{tabular}{lccc}
\hline
Simulation 2 Method & Mean FDP (SD) & Mean Power (SD) & $n_{\mathrm{rejected}}$ \\
\hline
HA-PRx          & 0.08 (0.01) & 0.70 (0.02) & 376.5 \\
Two-fold PRx    & 0.07 (0.01) & 0.69 (0.03) & 370.4 \\
PREM            & 0.12 (0.03) & 0.69 (0.03) & 390.9 \\
AdaPT-GMM          & 0.08 (0.02) & 0.56 (0.02) & 306.5 \\
ZAP             & 0.07 (0.01) & 0.68 (0.02) & 364.4 \\
CAMT            & 0.06 (0.02) & 0.61 (0.03) & 322.6 \\
IHW-BH          & 0.05 (0.01) & 0.56 (0.03) & 295.2 \\
Efron-2grp      & 0.07 (0.02) & 0.57 (0.03) & 308.2 \\
BH              & 0.05 (0.01) & 0.51 (0.03) & 269.9 \\
Oracle          & 0.10 (0.01) & 0.74 (0.02) & 411.2 \\
\hline
\end{tabular}
\begin{tabular}{lccc}
\hline
Simulation 3 Method & Mean FDP (SD) & Mean Power (SD) & $n_{\mathrm{rejected}}$ \\
\hline
HA-PRx          & 0.10 (0.02) & 0.92 (0.01) & 513.2 \\
Two-fold PRx    & 0.10 (0.02) & 0.92 (0.01) & 511.9 \\
PREM            & 0.12 (0.02) & 0.92 (0.01) & 526.0 \\
AdaPT-GMM          & 0.10 (0.01) & 0.87 (0.01) & 480.8 \\
ZAP             & 0.10 (0.02) & 0.92 (0.02) & 506.4 \\
CAMT            & 0.09 (0.02) & 0.85 (0.02) & 466.3 \\
IHW-BH          & 0.06 (0.01) & 0.82 (0.02) & 434.5 \\
Efron-2grp      & 0.05 (0.01) & 0.84 (0.02) & 444.7 \\
BH              & 0.05 (0.01) & 0.78 (0.02) & 408.7 \\
Oracle          & 0.10 (0.01) & 0.92 (0.01) & 509.8 \\
\hline
\end{tabular}

\caption{Simulation results --- the mean FDP, mean power, and mean rejections are calculated over thirty replicates of the data. Parentheses display standard deviations.}
\label{tab:simulation-bayesian}
\end{table}

In the first three simulations HA-PRx and two-fold PRx are competitive with the oracle and generally attain higher power than the non-oracle rules while maintaining BFDR control. In particular, two-fold PRx maintains power that is slightly less than that of HA-PRx, but this comes at the added benefit of thresholding BFDR more conservatively. We note that their average false discovery proportions tend to be greater than that of competing methods while remaining thresholded at our target level $\alpha=0.1$. This suggests that learning covariate information is pushing PRx to spend more of its BFDR budget, which is a desirable property since a power-optimal rule would attain $\text{BFDR}\approx\alpha$. A notable aspect of the PREM rejection rule is that it seems to consistently over-reject hypotheses, attaining an average false discovery proportion higher than $0.1$, but this is not accompanied by a significant increase in power. The reason for this can be seen in Figure \ref{fig: altdensity-bayesian}; notice that, while being able to capture shifts in the null proportion, methods such as PREM or Efron's 2-group procedure rely on a fixed, bimodal alternative density estimate that does not adapt with $x$. Methods such as IHW-BH or CAMT, which rely on two-sided $p$-values, are similarly disadvantaged through a loss of information in discarding the sign of $z$. PRx, on the other hand, captures the covariate-dependency in the alternative density, and is thus able to recover the directionality of the $z$-values as they vary over the covariate space. This is especially important given the setup of Simulations 1 and 2, which concentrates alternative test statistics at negative values for $x<0.5$ and at positive values for $x>0.5$; by ignoring this $x$-dependency, many competing procedures end up spuriously rejecting test statistics that are opposite to where the alternative density would place them.

\begin{figure}[ht]
\centering
\begin{minipage}{\textwidth} 
  \centering
  \includegraphics[width=0.8\textwidth]{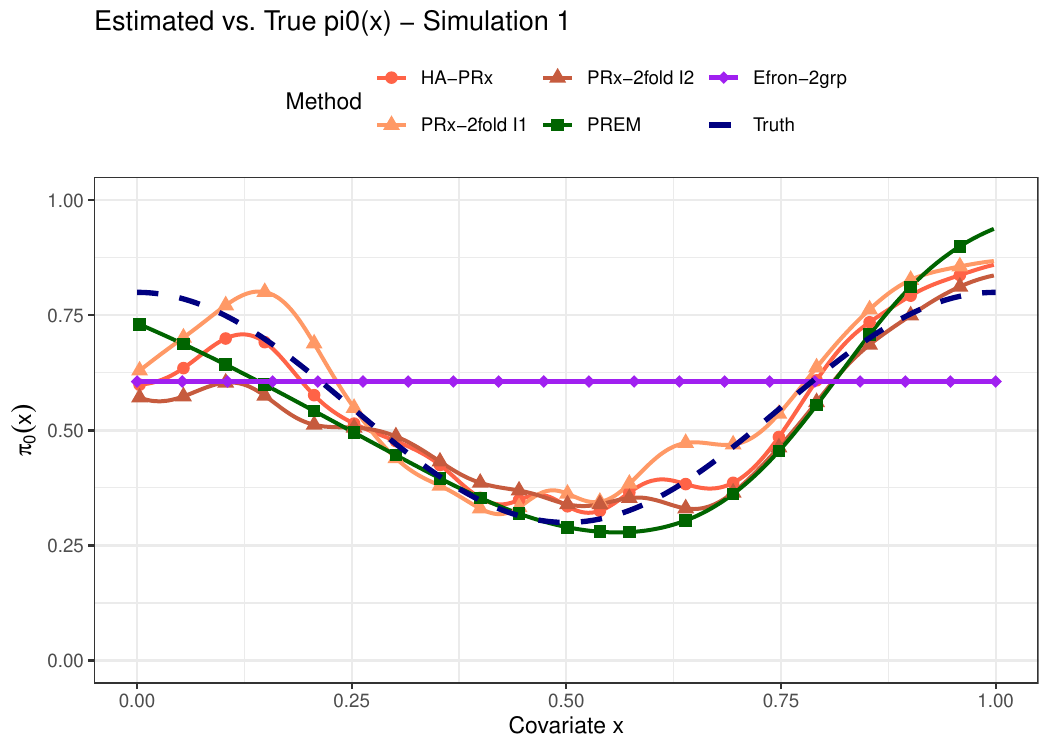}\\
\end{minipage}
\begin{minipage}{\textwidth} 
  \centering
  \includegraphics[width=0.8\textwidth]{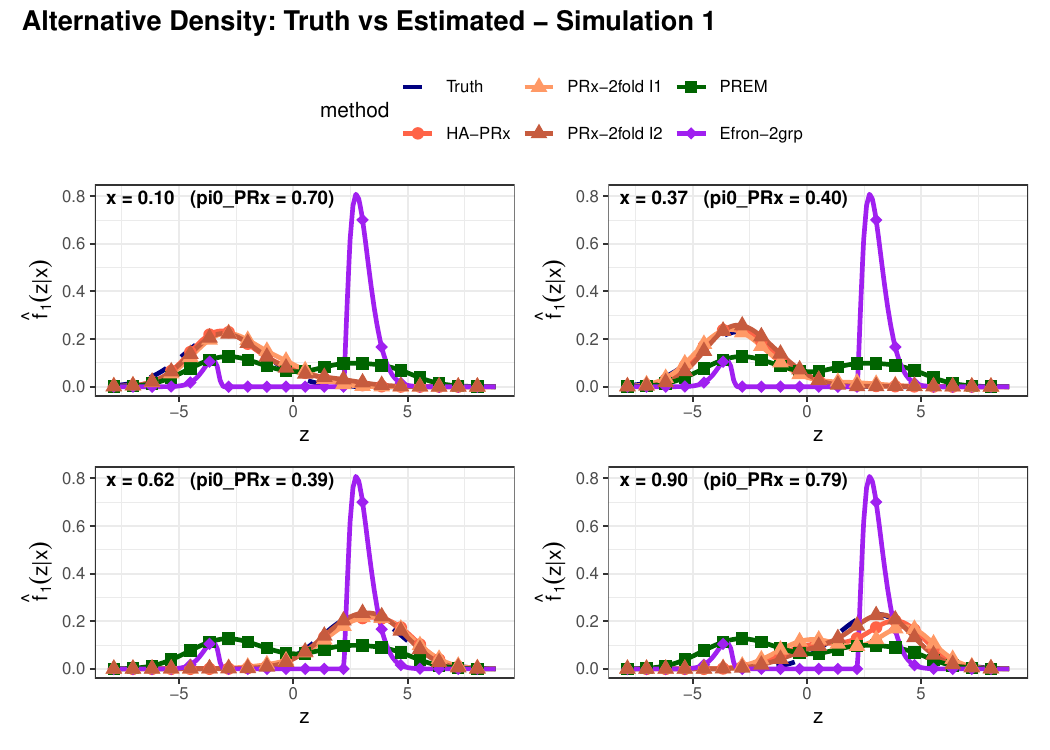}\\
\end{minipage}
\caption{Top panel displays the true underlying null proportion, $\pi_0(x)$ (dashed line), plotted against estimates $\hat{\pi}_0(x)$ given by HA-PRx (circle), fold 1 and fold 2 of two-fold PRx (triangles), PREM (square), and Efron's 2-group procedure (diamond). Bottom panel displays alternative density estimates. This is for simulation 1 under the Bayesian scheme. PRx-based methods admit a covariate-adaptive alternative density, while the others give only an $x$-agnostic density estimate.}
\label{fig: altdensity-bayesian}
\end{figure}

It should be noted that in the above simulation study, each of the $30$ replicates re-simulated both the latent signal indicators $H_i|x_i$ as well as the data $z_i|x_i$. This follows a Bayesian interpretation of FDR control, in which the indicators $H_i$ are treated as random quantities that share some common conditional distribution. In this sense we may think of the $\text{Bernoulli}(1-\pi_0(x))$ distribution as a ``conditional prior'' on the latent variables $H_i$, which, in the case of the simulation, we take to be a real mechanism of the data-generating process. A skeptical frequentist might disagree with the specification of a particular prior, and would instead be interested in controlling the false discovery rate while fixing an arbitrary sequence $\{H_i\}_{i=1}^n$ of hypotheses. We investigate this with a secondary set of simulations. In these frequentist simulations, we begin by fixing a $0-1$ propensity function, $\pi_0(x) = \mathbbm{1}\{x\in[0,0.1]\cup[0.25,0.45]\cup[0.55, 0.8]\cup [0.9,1]\}$ and generating a single sequence of $n$ covariates, $\{x_i\}_{i=1}^n $. We run $B=30$ replicates, each with $n=1000$ hypotheses, at the control level $\alpha = 0.1$. In each replicate we re-simulate $z$-values according to the appropriate null and alternative hypotheses: $z_i|x_i,H_i=0 \sim N(0,1)$ and $z_i|x_i,H_i=1\sim f_1(\cdot|x_i)$. We use the alternative densities of the earlier simulations for this. The results are presented below, in Table \ref{tab:simulation-frequentist}.

\begin{table}[ht]
\centering

\begin{tabular}{lccc}
\hline
Simulation 1 Method & Mean FDP (SD) & Mean Power (SD) & $n_{\mathrm{rejected}}$ \\
\hline
HA-PRx            & 0.07 (0.02) & 0.95 (0.02) & 352.6 \\
Two-fold PRx      & 0.05 (0.02) & 0.94 (0.02) & 343.2 \\
PREM              & 0.11 (0.02) & 0.73 (0.03) & 282.6 \\
AdaPT-GMM            & 0.10 (0.02) & 0.72 (0.04) & 275.6 \\
ZAP               & 0.10 (0.02) & 0.74 (0.04) & 284.7 \\
CAMT              & 0.09 (0.02) & 0.71 (0.03) & 272.5 \\
IHW-BH            & 0.08 (0.02) & 0.58 (0.05) & 215.7 \\
Efron-2grp        & 0.10 (0.02) & 0.72 (0.03) & 275.5 \\
BH                & 0.07 (0.02) & 0.68 (0.03) & 249.4 \\
Oracle            & 0.10 (0.00) & 1.00 (0.00) & 383.0 \\
\hline
\end{tabular}
\begin{tabular}{lccc}
\hline
Simulation 2 Method & Mean FDP (SD) & Mean Power (SD) & $n_{\mathrm{rejected}}$ \\
\hline
HA-PRx            & 0.03 (0.01) & 0.81 (0.04) & 288.9 \\
Two-fold PRx      & 0.03 (0.01) & 0.78 (0.04) & 276.9 \\
PREM              & 0.11 (0.02) & 0.53 (0.03) & 204.1 \\
AdaPT-GMM             & 0.09 (0.02) & 0.51 (0.04) & 193.9 \\
ZAP               & 0.09 (0.02) & 0.56 (0.04) & 214.7 \\
CAMT              & 0.09 (0.02) & 0.52 (0.04) & 195.9 \\
IHW-BH            & 0.06 (0.02) & 0.45 (0.03) & 165.7 \\
Efron-2grp        & 0.08 (0.02) & 0.51 (0.03) & 192.5 \\
BH                & 0.06 (0.02) & 0.48 (0.03) & 174.6 \\
Oracle            & 0.10 (0.00) & 1.00 (0.00) & 383.0 \\
\hline
\end{tabular}
\begin{tabular}{lccc}
\hline
Simulation 3 Method & Mean FDP (SD) & Mean Power (SD) & $n_{\mathrm{rejected}}$ \\
\hline
HA-PRx            & 0.10 (0.01) & 0.99 (0.01) & 379.1 \\
Two-fold PRx      & 0.09 (0.01) & 0.96 (0.05) & 363.4 \\
PREM              & 0.11 (0.02) & 0.83 (0.02) & 323.9 \\
AdaPT-GMM             & 0.10 (0.02) & 0.85 (0.02) & 325.7 \\
ZAP               & 0.10 (0.02) & 0.78 (0.06) & 300.7 \\
CAMT              & 0.10 (0.02) & 0.76 (0.03) & 289.5 \\
IHW-BH            & 0.08 (0.02) & 0.60 (0.03) & 227.0 \\
Efron-2grp        & 0.08 (0.01) & 0.81 (0.03) & 301.8 \\
BH                & 0.06 (0.01) & 0.71 (0.04) & 260.6 \\
Oracle            & 0.10 (0.00) & 1.00 (0.00) & 383.0 \\
\hline
\end{tabular}

\caption{Results for frequentist simulations. Similar to Table \ref{tab:simulation-bayesian}, the reported mean FDP and mean power are calculated over thirty replicates of the data. Parentheses display standard deviations.}
\label{tab:simulation-frequentist}
\end{table}

In the frequentist setting we continue to see HA-PRx and two-fold PRx leading in terms of power. It should be noted that HA-PRx fails to control FDR in the third simulation study --- but the additional false discovery rate accrued as a result of this is miniscule, and is accompanied by a significant leap in power. Furthermore, two-fold PRx successfully thresholds FDR in all simulations at the cost of a minor dip in power, while still maintaining a higher power than all non-PRx methods. The reason for this is the flexibility of PRx; PREM's approximation of $\pi_0(x)$, as one example, misses its discrete jumps as a result of underestimating the degrees of freedom in its spline-based estimate. The null proportion recovered by PRx, as seen in Figure \ref{fig:pi0 - alternative - frequentist}, manages to approximate each jump---even if its estimate is overly smoothed as a result of the Gaussian localization kernel---while also recovering the covariate-dependency in the signals. How close the realized mean false discovery proportion in simulation 3 is to the threshold level $\alpha$ suggests that, by virtue of PRx being an \textit{empirical Bayes} procedure as opposed to a fully Bayes procedure, it is able to approximate frequentist error rate control properties in large sample settings.

\begin{figure}[ht]
\centering
\begin{minipage}{\textwidth}
  \centering
  \includegraphics[width=0.8\textwidth]{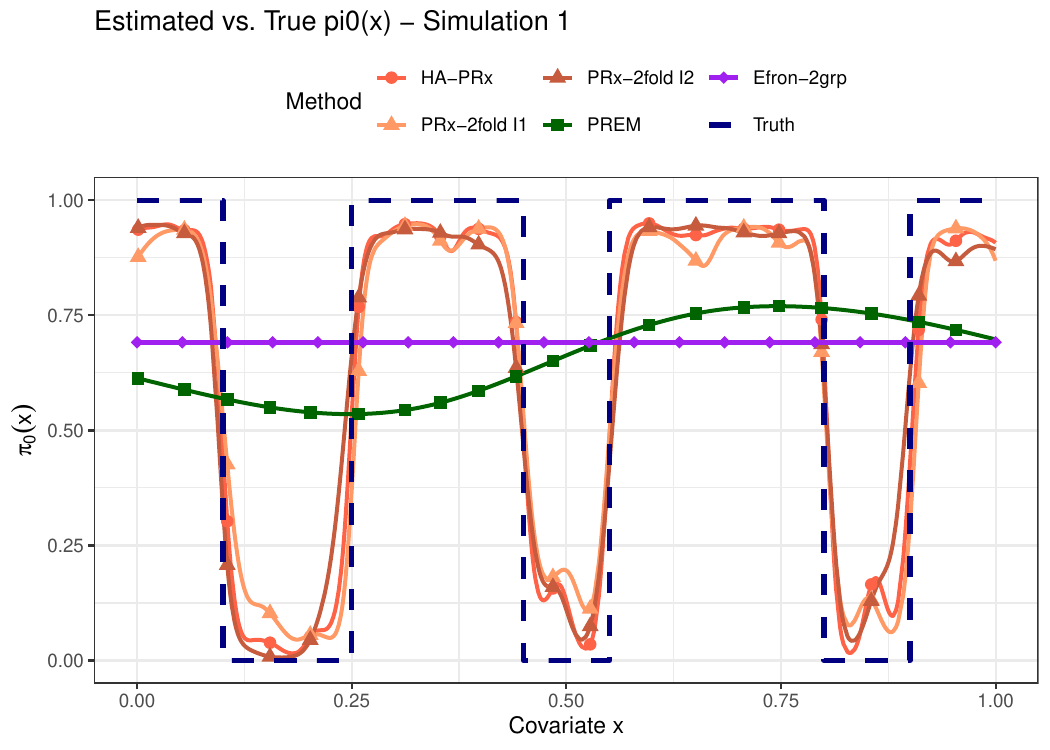}\\
\end{minipage}
\begin{minipage}{\textwidth}
  \centering
  \includegraphics[width=0.8\textwidth]{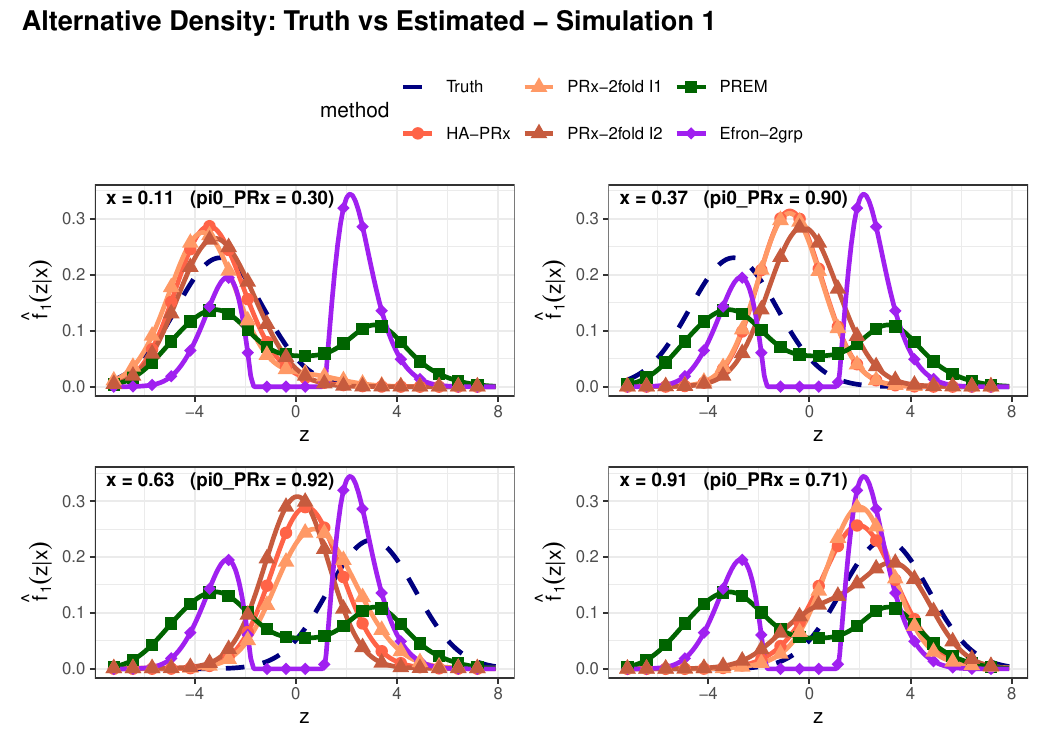}\\
\end{minipage}
\caption{Top panel displays the true underlying null proportion, $\pi_0(x)$ (dashed line), plotted against estimates $\hat{\pi}_0(x)$ given by HA-PRx (circle), fold 1 and fold 2 of two-fold PRx (triangles), PREM (square), and Efron's 2-group procedure (diamond). Bottom plot displays the alternative density estimates with the same shape scheme. This is for Simulation 1 under the frequentist scheme. Notice that PRx is able to adapt fairly well to the 0-1 null proportion. }

\label{fig:pi0 - alternative - frequentist}
\end{figure}

We conclude this section with a simulation involving a two-dimensional covariate. Motivation for this regime is given by spatially-aware data, in which we might observe our test statistics in tandem with coordinates, and from this map our hypotheses onto a two-dimensional plane. For this, let $X = (x_{1},x_{2})$, with $x_{i}\sim \text{Unif}[0,1]$ i.i.d. Let $\pi_0(x_1,x_2) = 0.9 - 0.5\cdot\sin(\pi x_1)\cdot\sin(\pi x_2)$. Have $H_i|X_i\sim \text{Bernoulli}(1-\pi_0(X_i))$, and let $z_i|H_i=0 \sim N(0,1)$ and $z_i|H_i=1,x_i\sim N(\theta_i, 1)$, with $\theta_i|X_i\sim N(\mu(X_i), 0.8^2)$, with $\mu(X) = -4 + 2\cdot(x_1 + x_2)$ for $x_1 + x_2 <1$ and $\mu(X) = 2\cdot (x_1+x_2)$ for $x_1+x_2\geq 1$. We also run an analogous frequentist simulation by letting $\pi_0(x_1,x_2) = 1-1\{\exists k\in\{1,2,3\}: (x_1 - c_{k,1}^2 ) + (x_2 - c_{k,2})^2 \leq r^2\}$ for $r = 0.15$, $c_1 = (0.2,0.2), c_2 = (0.5,0.7)$, and $c_3 = (0.8,0.4)$. Both of these simulations are repeated $B=30$ times, each with $n=1000$ hypotheses and a target control level of $\alpha=0.1$. The data is displayed in Figure \ref{fig:data hd} and the results are given in Table \ref{tab:simulation-hd} and Figures \ref{fig:data - alternative - hd - bayesian} and \ref{fig:data - alternative - hd - frequentist}.

As with all the other simulations we have done, methods based on PRx beat all non-oracle rules in terms of power, and by a significant amount. While HA-PRx fails to precisely control BFDR, the additional error accrued is miniscule in comparison to the subsequent increase in power. Furthermore, we find that two-fold PRx successfully controls BFDR, while also attaining a higher power than all non-PRx procedures---again, as a result of its ability to adapt to both a null proportion surface, and to a covariate-dependent alternative density.

\begin{figure}[ht]
\centering
\begin{minipage}{\textwidth}
  \centering
  \includegraphics[width=0.9\textwidth]{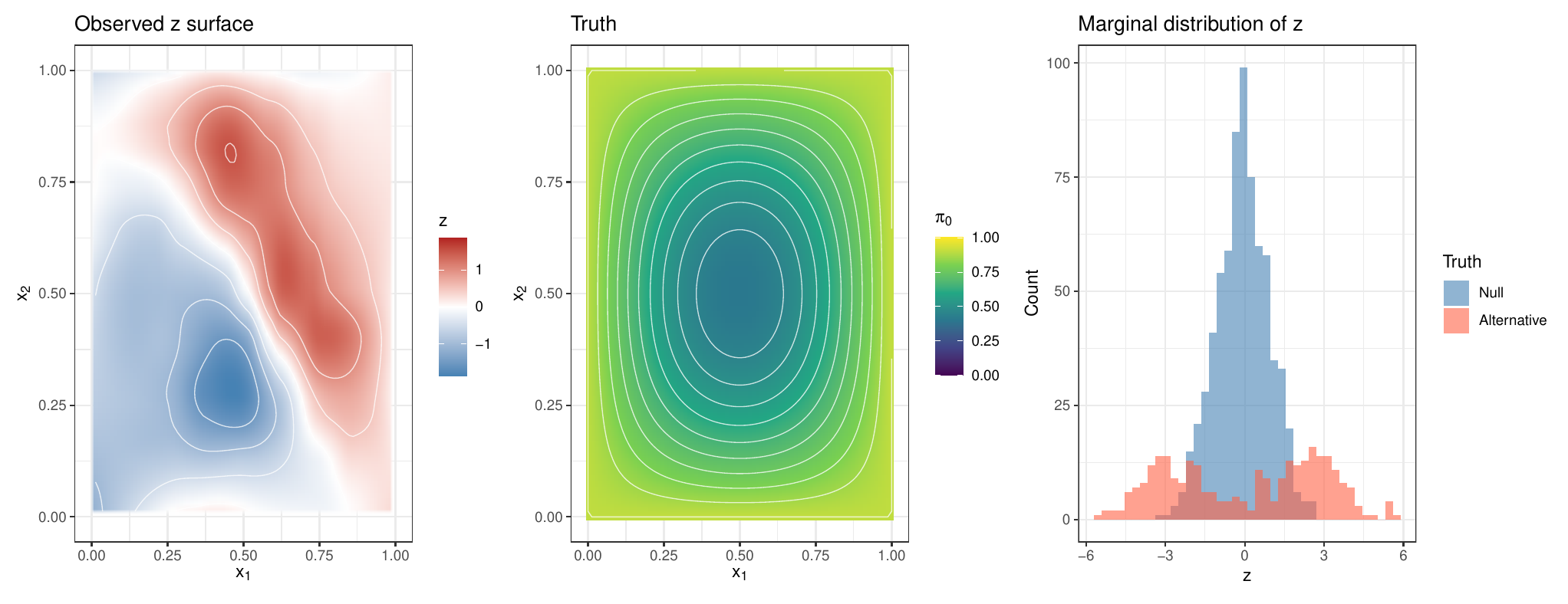}\\
\end{minipage}
\begin{minipage}{\textwidth}
  \centering
  \includegraphics[width=0.9\textwidth]{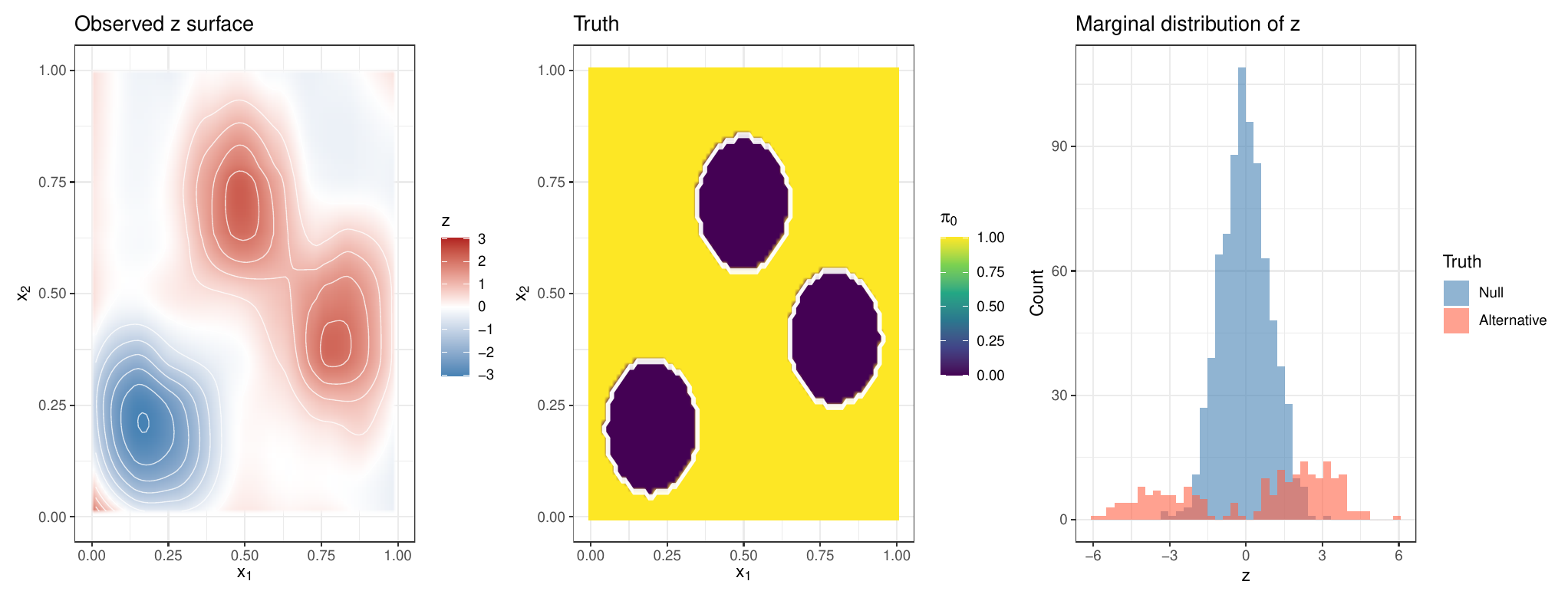}\\
\end{minipage}
\caption{From left to right: (1) topographic map of the $z$-values over the covariate space, (2) topographic map of the true null proportion, $\pi_0(x_1,x_2)$, and (3) a histogram of observed $z$-values. Top panel displays the data from the Bayesian 2D simulation, while the bottom panel displays the data from the frequentist 2D simulation.}
\label{fig:data hd}
\end{figure}

\begin{figure}[ht]
\centering
\begin{minipage}{\textwidth}
  \centering
  \includegraphics[width=0.8\textwidth]{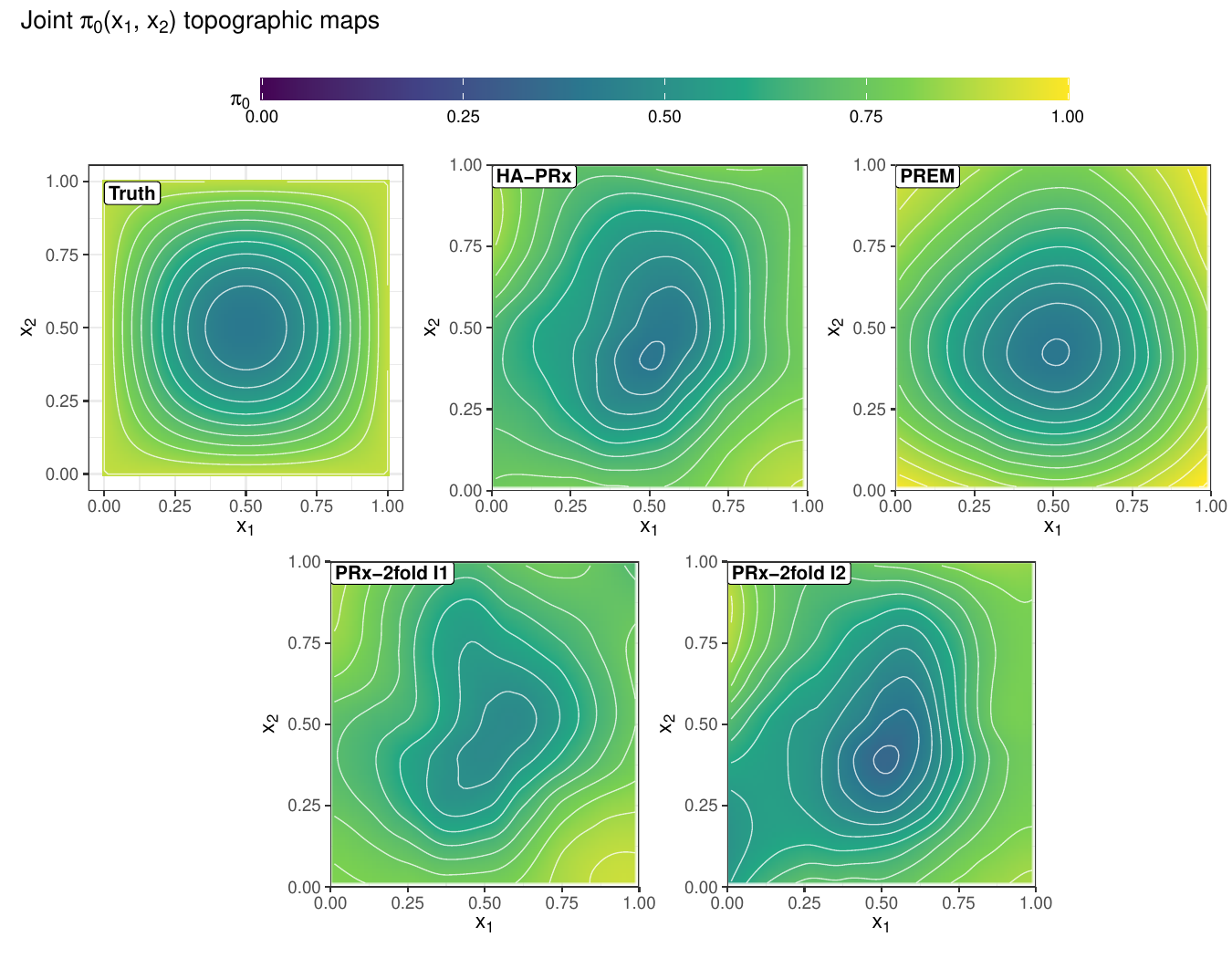}\\
\end{minipage}
\begin{minipage}{\textwidth}
  \centering
  \includegraphics[width=0.8\textwidth]{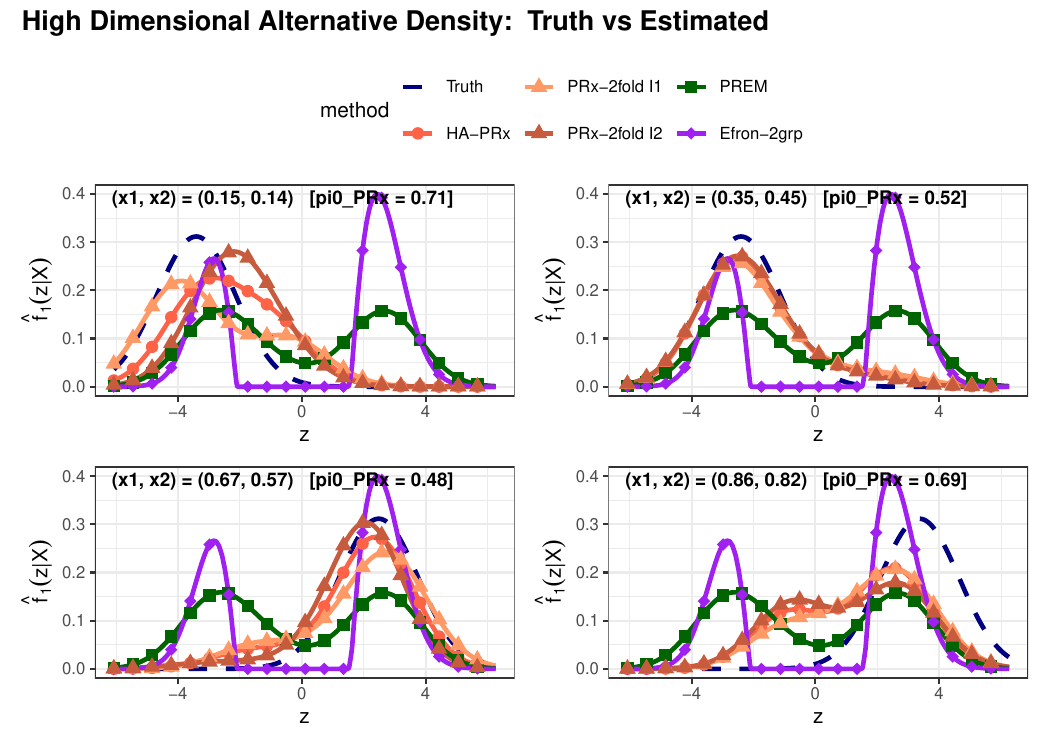}\\
\end{minipage}
\caption{Top panel displays topographic maps of the true null proportion for the 2D Bayesian simulation, as well as its estimates from PRx-based methods and PREM. The bottom plot displays the alternative density estimates with the same shape scheme as Figures \ref{fig: altdensity-bayesian} and \ref{fig:pi0 - alternative - frequentist}.}
\label{fig:data - alternative - hd - bayesian}
\end{figure}

\begin{figure}[ht]
\centering
\begin{minipage}{\textwidth}
  \centering
  \includegraphics[width=0.8\textwidth]{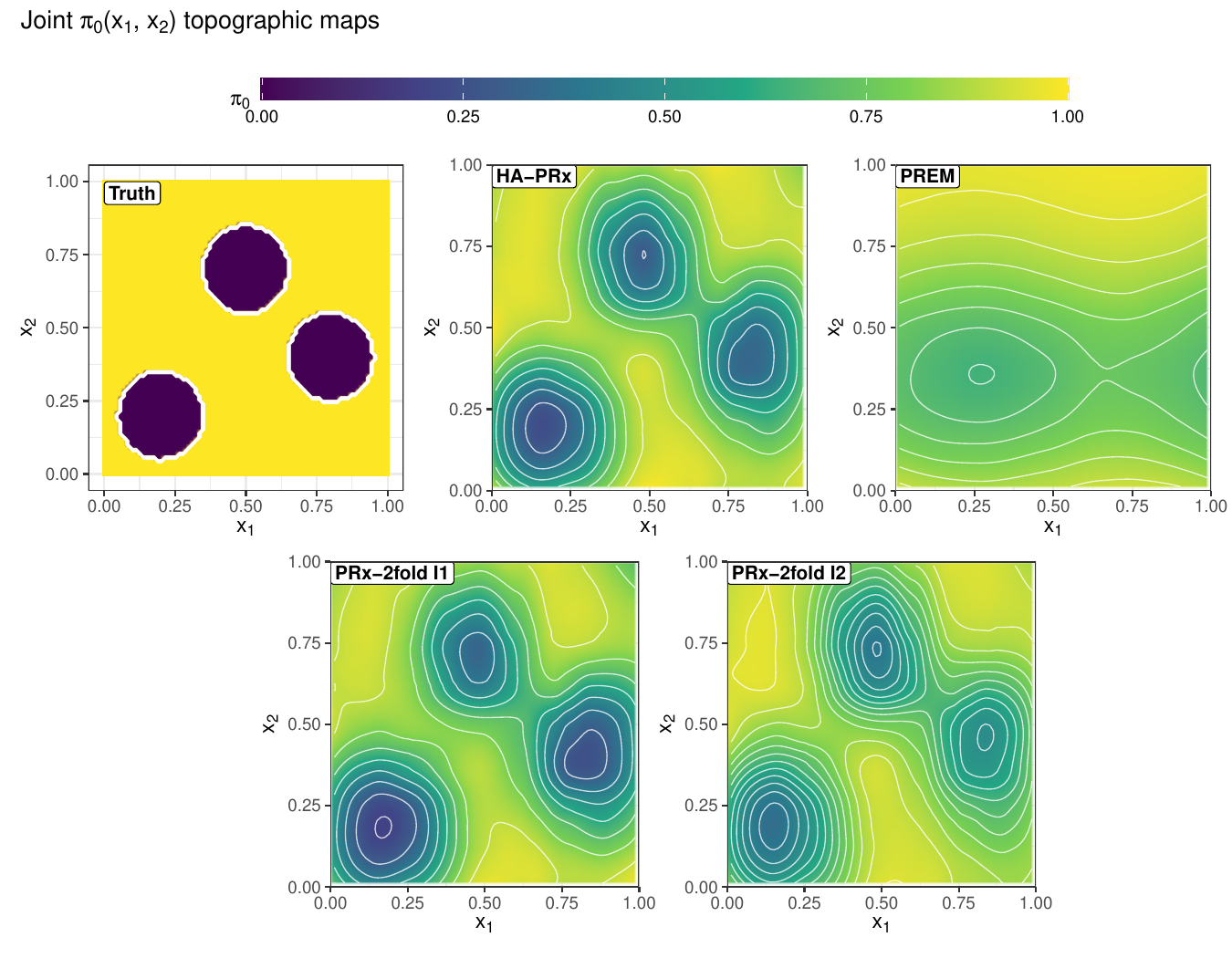}\\
\end{minipage}
\begin{minipage}{\textwidth}
  \centering
  \includegraphics[width=0.8\textwidth]{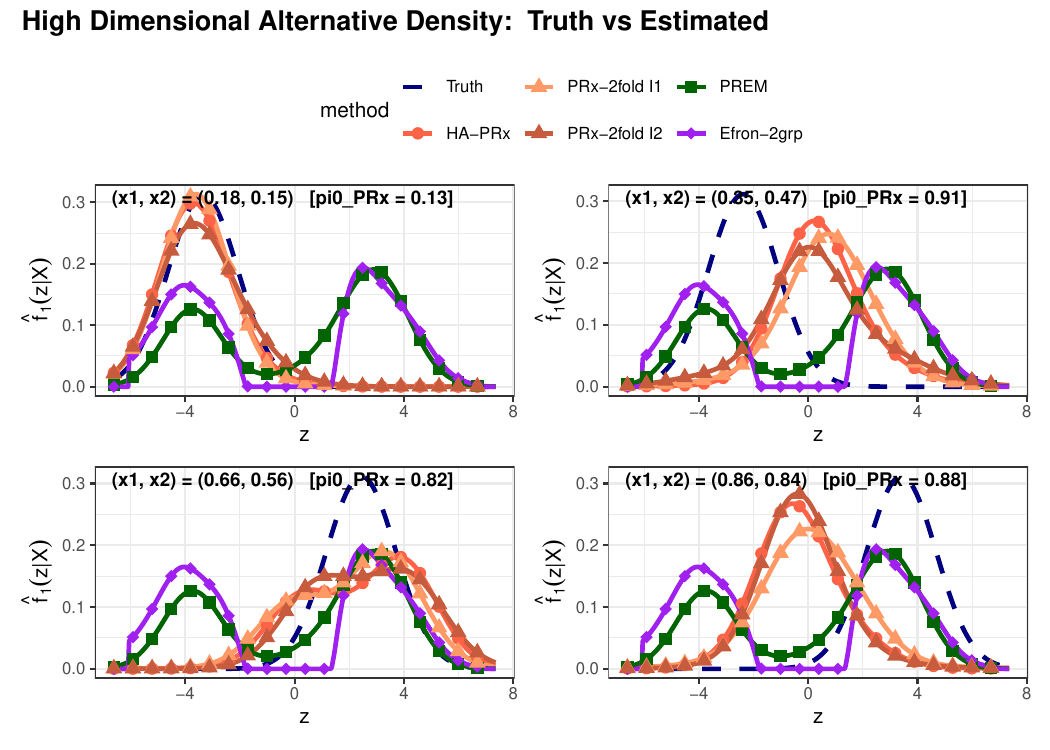}\\
\end{minipage}
\caption{Top panel displays topographic maps of the true null proportion for the 2D frequentist simulation, as well as its estimates from PRx-based methods and PREM. The Bottom panel displays the true alternative density at various points in the covariate space, plotted against PRx-based methods, PREM, and Efron's 2-groups procedure.}
\label{fig:data - alternative - hd - frequentist}
\end{figure}

\begin{table}[ht]
\centering

\begin{tabular}{lccc}
\hline
2D Simulation & Mean FDP (SD) & Mean Power (SD) & $n_{\mathrm{rejected}}$ \\
\hline
HA-PRx       & 0.11 (0.02) & 0.73 (0.03) & 247.6 \\
Two-fold PRx & 0.10 (0.02) & 0.70 (0.03) & 237.0 \\
PREM         & 0.10 (0.02) & 0.65 (0.04) & 217.0 \\
AdaPT-GMM        & 0.09 (0.03) & 0.54 (0.05) & 179.0 \\
ZAP          & 0.10 (0.03) & 0.64 (0.05) & 214.7 \\
CAMT         & 0.10 (0.03) & 0.58 (0.05) & 197.2 \\
IHW-BH       & 0.07 (0.02) & 0.52 (0.05) & 169.4 \\
Efron-2grp   & 0.09 (0.02) & 0.59 (0.03) & 196.9 \\
BH           & 0.07 (0.02) & 0.55 (0.03) & 179.2 \\
Oracle       & 0.10 (0.01) & 0.75 (0.03) & 251.7 \\
\hline
\end{tabular}

\begin{tabular}{lccc}
\hline
2D Frequentist & Mean FDP (SD) & Mean Power (SD) & $n_{\mathrm{rejected}}$ \\
\hline
HA-PRx       & 0.12 (0.02) & 0.90 (0.03) & 213.3 \\
Two-fold PRx & 0.09 (0.02) & 0.88 (0.03) & 199.7 \\
PREM         & 0.11 (0.03) & 0.62 (0.04) & 143.7 \\
AdaPT-GMM        & 0.09 (0.03) & 0.76 (0.04) & 174.0 \\
ZAP          & 0.11 (0.04) & 0.63 (0.05) & 145.8 \\
CAMT         & 0.10 (0.04) & 0.52 (0.06) & 120.2 \\
IHW-BH       & 0.08 (0.03) & 0.44 (0.04) & 100.1 \\
Efron-2grp   & 0.10 (0.02) & 0.58 (0.03) & 133.2 \\
BH           & 0.08 (0.02) & 0.55 (0.04) & 122.6 \\
Oracle       & 0.10 (0.00) & 1.00 (0.00) & 230.0 \\
\hline
\end{tabular}

\caption{Results for the 2D simulation.}
\label{tab:simulation-hd}
\end{table}

\subsection{Neural Synchrony}\label{sec:neural}
Neural synchrony refers to the phenomenon in which two or more neurons exhibit electrical spikes in very quick succession or at the same time. It is believed that this phenomenon plays an integral role in the proper functioning of the brain, and certain cognitive disorders have been linked to a disruption of it. We make use of the \textit{synchrony statistic}, $\hat{\zeta}_i$, to quantify the synchrony between a pair of neurons. Detecting the presence of neural synchrony given a collection of recorded neurons has proven a difficult task; applying the Benjamini-Hochberg procedure to a collection of ${n\choose 2}$ hypotheses ignores auxiliary information that we might have on the neurons. It therefore becomes desirable to use a covariate-adaptive multiple testing procedure to solve this problem.

The dataset we consider is publicly available at $\texttt{https://doi.org/10.7910/DVN/WJABUK}$. It consists of a large number of synchrony statistics $\hat{\zeta}_i$ accompanied by $s_i$, the estimated standard error of $\log \hat{\zeta}_i$. Following the work of \cite{smith2008spatial, kelly2010local}, we can calculate $z_i = \log \hat{\zeta}_i/s_i$, and treat $\{z_i\}_{i=1}^n$ as our collection of test statistics. We consider two covariates of interest: the distance between the neurons in micrometers, and the correlation between the neurons' tuning curves. Both of these are hypothesized to affect neural synchrony. The work of \cite{scott2015false} uses PREM to detect significant $z$-values, and found that doing so leads to more rejections than what is achieved with a covariate-free application of Efron's two-groups model, in which the $z$-values are modeled as being generated from $m(z) = \pi _0 f_0(z) + (1-\pi_0)f_1(z)$. In using PREM, the two-groups model is implicitly extended to include $\pi_0(x)$ such that $m(z|x) = \pi_0(x)f_0(z) + (1-\pi_0(x))f_1(z)$. We propose modeling $z_i$ using our version of the covariate-dependent two-groups model with a null distribution of $N(\theta_0, \sigma_0^2)$, a covariate-dependent null proportion $\pi_0(x)$, and a covariate-dependent alternative density, $f_1(z|x)$. Two-fold PRx is used to discover any significant $z$-values which we tentatively label as a synchronous pair. We use PREM as a reference procedure. Our goal in this is to emphasize the benefit gained by using PRx to recover the local false discovery rate surface; being the only other procedure that shares this aim, PREM offers a valuable benchmark. 

The methodology involved in this analysis differs from that of the simulation studies through the fact that we now wish to estimate an \textit{empirical null}, $N(\theta_0, \sigma_0^2)$, where $\theta_0$ and $\sigma_0^2$ are not assumed to be known. An empirical null is useful to estimate when classical testing assumptions are violated. For example, it has been shown that correlation between the $z$-values may disperse or concentrate the null component relative to the standard Gaussian \citep{efron2007correlation}. This is especially relevant in this application --- since each test corresponds to a neuron pair, it is reasonable to expect correlation between two test statistics involving a common neuron. Furthermore, the accumulation of signals into a central bulk is a well-documented phenomenon when signals are permitted to lie arbitrarily close to $z=0$ \citep{xiang2024interpretation}. Even if the true null is a standard Gaussian, the observed central bulk may become warped as a result of these near-null signals. Recovery of $N(\theta_0,\sigma_0^2)$ in place of $N(0,1)$ adjusts for this.

The dataset we use consists of $n=7004$ pairs, each of which comes with a test statistic $z_i$ and two predictors, $x_1$ and $x_2$. These represent tuning curve correlation and inter-neuron distance, respectively. As a preprocessing step we use min-max normalization to rescale $\mathcal{X}$ so it fits in the unit ball $[0,1]^2$. PRMLx maximization is used to obtain estimates $\hat{b}_{\text{dist}},\hat{b}_{\text{corr}},\hat{\theta}_0,$ and $\hat{\sigma}_0$. Since we are using two-fold PRx, this must be iterated twice --- once over each fold. Maximization yields an estimate of $\hat{\theta}_0 = \{0.48,0.51\}$ and $\hat{\sigma}_0 = \{0.73,0.75\}$, suggesting positive bias and underdispersion in the empirical null relative to the theoretical null, $N(0,1)$. The bandwidth estimates are also interpretable---$\hat{b}_{\text{dist}} = \{48.59, 56.78\}$ and $\hat{b}_{\text{corr}} = \{25.10, 39.49\}$---suggesting more localized information-sharing over inter-neuron distances than over tuning curve correlations. Plugging these into our variation of the two-groups model, we run two-fold PRx to obtain estimates $\{\hat{\ell}_i\}_{i=1}^n$. We end up rejecting $1112$ hypotheses in comparison to $721$ from PREM. The rejections are visualized in Figure \ref{fig:disagreements}, and plotted over the covariate plane in Figure \ref{fig: covariates-against-rejections}. 

\begin{figure}[H]
    \centering
    \begin{minipage}{0.5\textwidth}
        \centering
        \includegraphics[width=\linewidth]{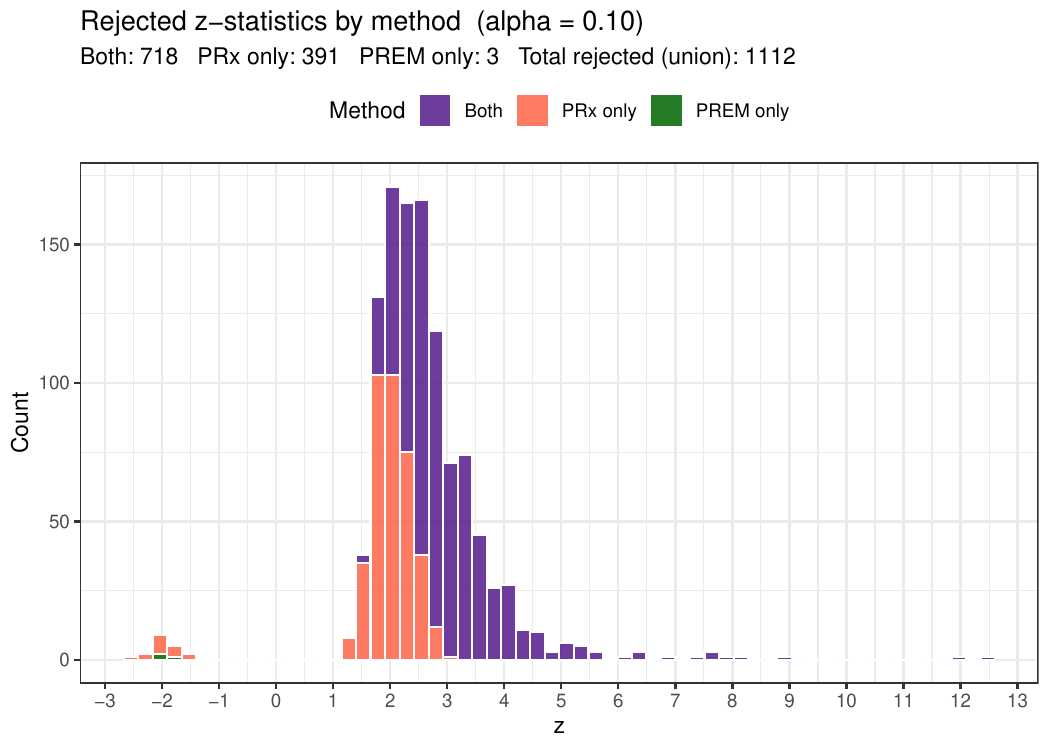}
    \end{minipage}
    \caption{Rejected $z$-values, colored by which method carried out the rejection.}
    \label{fig:disagreements}
\end{figure}

While it is impossible to say whether or not our rejections are \textit{better} than those obtained by PREM, it is clear from Figure \ref{fig: covariates-against-rejections} that the disagreements between the two methods, the vast majority of which consist of rejections made by PRx where PREM fails to reject, concentrate on a specific part of the covariate space---in particular, on smaller inter-neuron distances. In comparison, rejections from both procedures are spread out much more evenly over the tuning curve correlation axis. This suggests two things: firstly, that smaller inter-neuron distances do indeed correspond with a heightened chance of synchrony, and secondly, that PRx is acting on this much more strongly than PREM. Indeed, while both two-fold PRx and PREM reject heavily for small inter-neuron distances, the quantity of extra rejections made by PRx indicates that it is more heavily pushed by that covariate region to reject. 
\begin{figure}[H]
    \centering
    \begin{minipage}{0.48\textwidth}
        \centering
        \includegraphics[width=\linewidth]{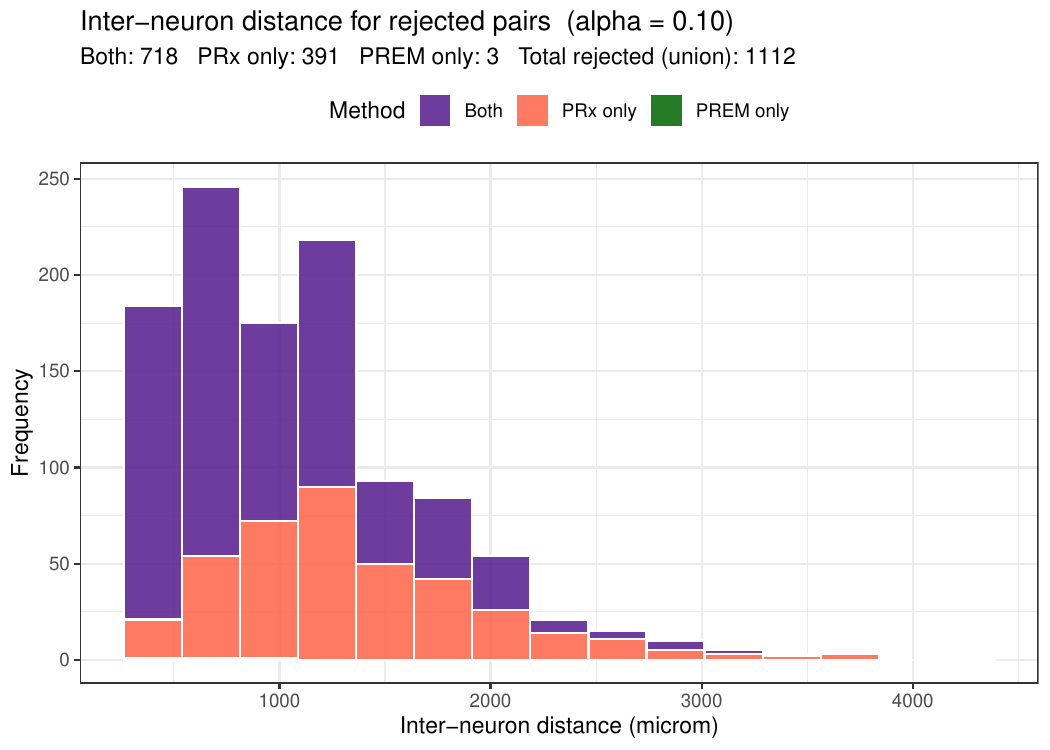}
    \end{minipage}
    \hfill
    \begin{minipage}{0.48\textwidth}
        \centering
        \includegraphics[width=\linewidth]{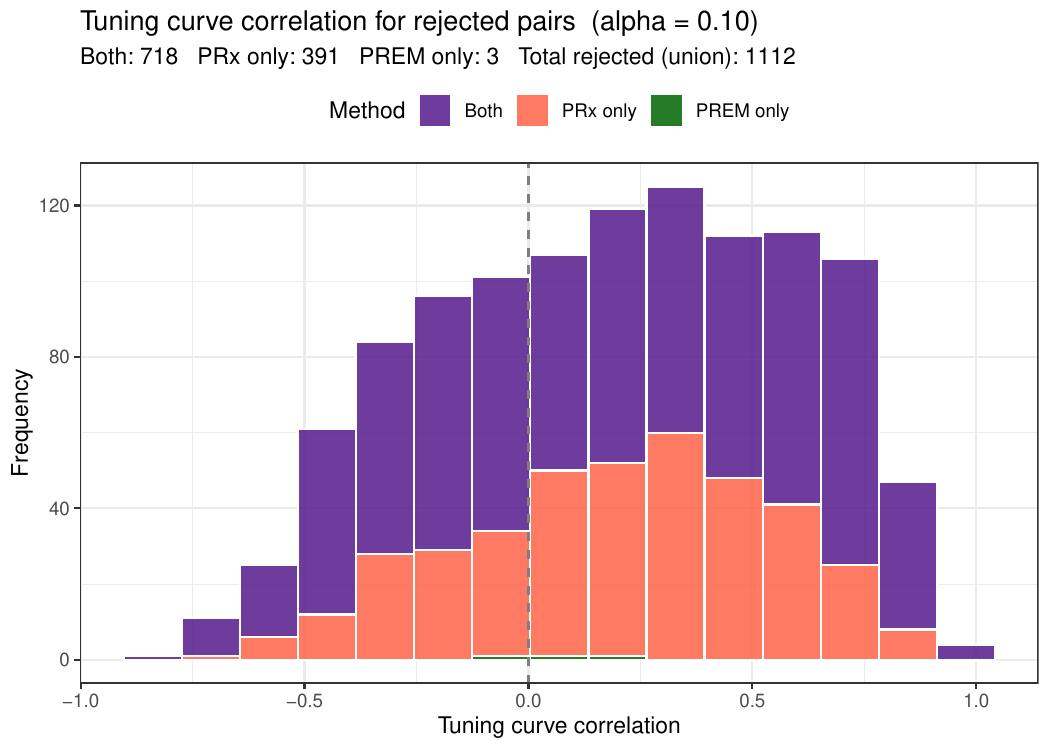}
    \end{minipage}
    
    \caption{Left panel displays rejection agreements/disagreements plotted over inter-neuron distances. Right panel displays rejection agreements/disagreements plotted over tuning curve correlations.}
    \label{fig: covariates-against-rejections}
\end{figure}
A notable aspect of the rejections made by PRx is that they cover nearly all rejections made by PREM; in this sense, PRx is the more aggressive procedure. We furthermore find through Figure \ref{fig:disagreements} that the majority of rejections made only by PRx are concentrated on the left side of the rejected bulk, corresponding to $z$-values closer to the center of the empirical null. Thus, relative to PREM, PRx differs not only in the number of rejections it makes, but also in their location. In particular, PRx identifies additional hypotheses at more moderate $z$-values, rather than concentrating its rejections at only the most extreme observations. These additional rejections may be of interest in scientific applications, particularly when practitioners have the opportunity to further investigate individual discoveries. While we cannot determine from this analysis alone whether these additional rejections represent true signals, the pattern suggests that PRx and PREM differ meaningfully in how they identify potentially informative near-null observations---and between the two, it is PRx that makes strictly more discoveries.

\section{Theoretical Results}\label{sec:theory}

We have seen in Section \ref{sec:casestudies} that PRx testing leads to highly competitive results in terms of both BFDR control and power maximization. What remains to be seen is whether or not it admits desirable theoretical properties. Towards this end we have two goals: to show that the PRx estimate of the localized false discovery rate is mathematically justified, and to show that the rejection rules we use admit provable error rate control properties. We begin by defining the PRx estimate. Let $\hat{\Psi}_n$ denote the $n$th PRx iterate. Letting $\ell(z,x) = \text{Pr}(H_i = 0|z,x)$ denote the oracle covariate-localized false discovery rate, we want to show that $\hat{\ell}_n(z,x)\rightarrow \ell(z,x)$, where $\hat{\ell}_n$ is the local false discovery rate estimate induced by $\hat{\Psi}_n$ under either naive, hole-adjusted, or conservative PRx. 

In order to accomplish this we require three main assumptions. {\bf (A1)} The oracle null proportion is bounded such that $0<\pi_{\text{min}}<\pi_0(x)<\pi_{\text{max}}<1$.  {\bf (A2)} The ground truth alternative mixing density $\psi(u|x)$ satisfies two-sided domination: $\kappa_1 \psi_0(u) \leq \psi(u|x)\leq\kappa_2 \psi_0(u)$ for constants $\kappa_1,\kappa_2$, and a density function $\psi_0(u)$. {\bf (A3)} There is a fixed $\delta > 0$ such that  $\psi(u|x) = 0$ on the interval $[-\delta, \delta]$. 

Assumptions (\textbf{A1})-(\textbf{A2}) are used in verifying the ratio condition in the PRx consistency theorem, which also requires the following additional regularity conditions. 
%the standard PRx consistency assumptions hold. 
Let us consider a model in which $z_i|x_i\sim \int \phi(z|u)\Psi(u|x)\mu(du)$ for some unknown mixing density $\Psi(u|x)$ and known, parametric kernel $\phi(z|u)$. Then under the following conditions, the PRx estimate $\hat{\Psi}_n$ converges almost surely in the weak topology to $\Psi$:
\begin{enumerate}
    \item The sequence $x_1,\ldots,x_n$ are i.i.d. and integrable. Furthermore, $\mathcal{X}$ and $\Theta$ are both compact.
    \item The mapping $u\rightarrow\phi(z|u)$ is bounded and continuous in $u$ for each fixed $y$.
    \item The mapping $f(u|x)\rightarrow\int\phi(z|u)f(u|x)\mu(du)$ is injective.
    \item For any $\epsilon>0$ and compact set $\mathcal{C}_0\subset\mathcal{Z}$, there exists a compact set $\mathcal{U}_0\subset\mathcal{U}$ such that $\int_{\mathcal{C}_0}\phi(z|u)\nu(dz)<\epsilon$ for all $\theta\notin\Theta_0$.
    \item There exists a constant $0<B<\infty$ such that for any $u_1,u_2,u_3\in\mathcal{U}$:
    $$\int_\mathcal{Y}\frac{\phi(z|u_1)^2}{\phi(z|u_2)^2}\phi(z|u_3)\nu(dz)<B$$
    \item There exists a constant $C>0$ such that for all $(x_i,z_i)$, it is the case that $m(z_i|x_i)\geq Cm(z_i|x)$. 
\end{enumerate}

In particular, assumption $6$ is the only one that requires control over $\pi_0(x)$ and $\psi(u|x)$, which (\textbf{A1}) and (\textbf{A2}) grant (see the supplementary); the other assumptions are fulfilled by the model laid out in Section \ref{sec:methods}. (\textbf{A3}) is a technical condition necessary to derive consistency of the estimated null proportion $\hat{\Psi}_n(\{0\}|x)$ to the oracle $\pi_0(x)$. With it, consistency arises naturally as a consequence of the Portmanteau theorem. It can also be interpreted as a further strengthening of the \textit{zero assumption} \citep{efron2004large}, in that it enforces a level of separability between the alternative and null distribution. In particular, it excludes the signal locations from a $\delta$-ball around $u=0$, hence pushing signals farther away from the central bulk. 

Our simulations in Section \ref{sec:sims} are not designed to follow assumption (\textbf{A3})---most involve a latent mixing density that places mass arbitrarily close to $0$. This represents a more difficult data regime than if it were satisfied; indeed, if the signal locations were truly bounded away from $u=0$, then they only become easier to distinguish from the null, and thus easier to reject. Despite this, the performance of hole-adjusted PRx in Section \ref{sec:sims} indicates a high degree of robustness towards violations of (\textbf{A3}), especially when $h$ is chosen to be extremely small. With these three assumptions, we may now present our main consistency result. 
%\begin{lemma}
%    Under assumptions (\textbf{A1}) and (\textbf{A2}), the PRx consistency conditions hold. 
%\end{lemma}
%\begin{proof}
%    By assumptions (\textbf{A1}) and (\textbf{A2}), and denoting $f_1^0(z) = \int_{-C}^C N(z;\theta_0 + u,\sigma_0^2)\psi_0(u)du$, one can show that $m(z|x_i)/m(z|x)$ satisfies the inequality:
%    \begin{equation}
%        \begin{split}
%            \frac{m(z|x_i)}{m(z|x)}\geq \frac{\pi_{\text{min}} f_0(z) + \kappa_1(1-\pi_{\text{max}})f_1^0(z)}{\pi_{\text{max}} f_0(z) + \kappa_2f_1^0(z)} = %\frac{\pi_{\text{min}}\tau(z) + \kappa_1(1-\pi_{\text{max}})}{\pi_{\text{max}}\tau(z) + \kappa_2}
%        \end{split}
%    \end{equation}
%    Where we denote $\tau(z) := f_0(z)/f_1^0(z)$. The final bound in the above expression is monotonic in $\tau$ --- in particular, if we write $g(\tau) = (\pi_{\text{min}}\tau +\kappa_1(1-\pi_{\text{max}})) / (\pi_{\text{max}}\tau + \kappa_2)$, one can show that $g(\tau)$ is monotonically increasing when $\pi_{\text{min}}\kappa_2>\kappa_1(1-\pi_{\text{max}})\pi_{\text{max}}$ and monotonically decreasing if it is the other way around. In either case, we now know that the infimum of $g(\tau)$ is achieved either at $\tau=0$ or at $\tau\rightarrow\infty$. This then gives us the explicit lower bound:
%    \begin{equation}
%        \begin{split}
%            \frac{m(z|x_i)}{m(z|x)}\geq \min\bigg(\frac{\pi_{\text{min}}}{\pi_{\text{max}}}, \frac{\kappa_1(1-\pi_{\text{max}})}{\kappa_2}\bigg)
%        \end{split}
%    \end{equation}
%\end{proof}

%\begin{proof}
%    See the supplementary.
%\end{proof}

\begin{theorem}[Consistency of PRx]\label{thm:consistency}
    Suppose assumptions (\textbf{A1}), (\textbf{A2}), and (\textbf{A3}) hold. In particular, let $\psi(\cdot|x) = 0$ on $[-\delta,\delta]$. Define $\hat{\ell}_n^{(h)}(z,x)$ to be the PRx estimate obtained by initializing $\hat{\Psi}_0^{(h)} (\cdot|x) = \hat{\pi}_{0,0}^{(h)}(x)\delta_0 + (1-\hat{\pi}_{0,0}^{(h)}(x))\hat{\psi}_0^{(h)}(\cdot|x)$, where $0<\hat{\pi}_{0,0}^{(h)}(x)<1$, $\hat{\psi}_0^{(h)}(\cdot|x) = 0$ on $[-h,h]$, and $\hat{\psi}_0(\cdot|x)>0$ everywhere else. For every $h\leq\delta$, we have:
    \begin{equation}
        \begin{split}
            \lim_{n\rightarrow\infty}\iint |\hat{\ell}_n^{(h)}(z,x) - \ell(z,x)| \Pi(dx)m(z|x)dz= 0\text{ a.s.}
        \end{split}
    \end{equation}
\end{theorem}
\begin{proof}
    Let us denote $\Theta = [-C,-\delta]\cup\{0\}\cup[\delta,C]$ and $\Theta_h := [-C,-h]\cup\{0\}\cup [h, C]$. Notice that $\Theta\subseteq\Theta_h$, and so we can say that both $\Psi$ and $\hat{\Psi}_n^{(h)}$ are probability measures supported on $\Theta_h$. Since $\Theta_h$ is compact and all the consistency assumptions of PRx hold on this expanded support, we have that $\hat{\Psi}_n^{(h)}(\cdot|x)\xrightarrow{w} \Psi(\cdot|x)$. Let $(z,x)$ be fixed points. Since $\{0\}$ is an isolated set in $\Theta_h$, applying Portmanteau theorem yields consistency of the estimated null proportion $\hat{\pi}_{0,n}^{(h)}(x)\xrightarrow{a.s.} \pi_0(x)$. By definition of weak convergence, we also have that $\hat{m}_n^{(h)}(z|x)\xrightarrow{a.s.} m(z|x)$. Continuous mapping theorem can then be used to show $\hat{\ell}_n^{(h)}(z,x)\xrightarrow{a.s.} \ell(z,x)$. To be clear, up to this point these convergence results are all with respect to the randomness in the data, with fixed $(z,x)$. Fubini's theorem, combined with the boundedness of $|\hat{\ell}_n^{(h)}(z,x)- \ell(z,x)|\leq 2$, then implies $L^1$ convergence over $(z,x)\sim \Pi(dx)m(z|x)dz$.
\end{proof}
A more ambitious goal than the above result would be full consistency: can we show that naive PRx leads to consistent recovery of the local false discovery rate? The primary difficulty in this latter approach comes in showing consistency of the estimated term $\hat{\Psi}_n(\{0\}|x)$ towards its oracle target, $\pi_0(x)$. Towards this end, Portmanteau may be leveraged to show that $\limsup_{n} \hat{\Psi}_n(\{0\}|x)\leq \pi_0(x)$. To bound the $\liminf$ term, it suffices to show that $\lim_{\epsilon\rightarrow0}\limsup_{n\rightarrow\infty}\int_{-\epsilon}^\epsilon\hat{\psi}_n(u|x)du=0$ --- but a proof of this final statement has so far eluded us. 

Assumption (\textbf{A3}) offers a second approach given by hole-adjusted PRx---initialize $\hat{\psi}_0(u|x) = 0$ on $u\in[-h,h]$, and set $h$ to be an extremely small value. As suggested by Theorem \ref{thm:consistency}, a sufficiently small value of $h$ leads to consistent recovery of the oracle. A third, more conservative approach is also available: using the fact that for any $\eta>0$, $[-\eta,\eta]$ is a continuity set, Portmanteau can be leveraged to argue that $\hat{\Psi}_n([-\eta,\eta]|x)\rightarrow \pi_0(x) + (1-\pi_0(x))\int_{-\eta}^\eta \psi(u|x)du$. We note that $(1-\pi_0(x))\int_{-\eta}^\eta\psi(u|x)du$ represents a bias from the true null proportion $\pi_0(x)$ that is a decreasing function in $\eta$. Hence, for $\eta$-small, we may take $\hat{\Psi}_n([-\eta,\eta]|x)$ as a conservative estimate of $\pi_0(x)$. The following lemma then holds:
\begin{lemma}[Shrinking-$\eta$ Consistency]\label{lem:conservativeconsistency}
    Let $0<\hat{\pi}_{0,0}(x)<1$, and denote the conservative local false discovery rate estimator to be $\hat{\ell}_{n,\eta}(z,x) := \min\Big\{\frac{\hat{\Psi}_n([-\eta,\eta]|x)f_0(z)}{\hat{m}_n(z|x)},1\Big\}$. Under assumptions (\textbf{A1}) and (\textbf{A2}), the following limiting result holds:
    \begin{equation}
        \begin{split}
            \lim_{\eta\downarrow 0}\lim_{n\rightarrow\infty} \iint |\hat{\ell}_{n,\eta}(z,x) - \ell(z,x)| \Pi(dx)m(z|x)dz= 0\text{ a.s.}
        \end{split}
    \end{equation}
\end{lemma}
\begin{proof}
    Note that the function $t\rightarrow\min(t,1)$ is continuous. From this we may argue using Portmanteau, combined with the continuous mapping theorem, that:
$$\hat{\ell}_{n,\eta}(z,x)\xrightarrow{n\rightarrow\infty} \frac{(\pi_0(x)+(1-\pi_0(x))\int _{-\eta}^\eta \psi(u|x)du)f_0(z)}{m(z|x)}\wedge1=: \ell_\eta(z,x)$$
Noting that $|\hat{\ell}_{n,\eta}-\ell(z,x)|<2$, we can apply dominated convergence theorem to conclude:
$$\iint |\hat{\ell}_{n,\eta}(z,x) - \ell(z,x)|\Pi(dx)m(z|x)dz \xrightarrow{n\rightarrow\infty}\iint |\hat{\ell}_{\eta}(z,x) - \ell(z,x)|\Pi(dx)m(z|x)dz$$
Since $\psi$ is a density, $\int_{\eta}^\eta\psi(u|x)du\xrightarrow{\eta\downarrow0}0$. This, combined with continuity of $\min(t,1)$, implies that $\ell_\eta(z,x)\rightarrow \ell(z,x)$. A second application of the dominated convergence theorem indexed over $\eta\downarrow 0$ then completes the proof.
\end{proof}
A slightly stronger variation of Theorem \ref{thm:consistency} and Lemma \ref{lem:conservativeconsistency} would index $\eta$ and $h$ by $n$, and then achieve a rate-controlled, limit-in-$n$ consistency result---but pursuing this leads to the same bottleneck as that of naive PRx. As it stands, conservative PRx overestimates the local false discovery rate by inflating its null proportion, but the asymptotic error accrued by this vanishes with smaller choices of $\eta$. Hole-adjusted and conservative PRx therefore represent two compromises for solving PRx testing; hole-adjusted PRx requires us to choose $h\leq\delta$ for some unknown $\delta$ --- but once we successfully choose such an $h$, we have consistency to the exact oracle target. Conservative PRx, by permitting the appearance of a bias term, never leads to consistency towards the oracle, but also does not necessitate $\eta$ to fall beneath a particular threshold, or for (\textbf{A3}) to hold at all. 

%It should be noted here that we do not provide theoretical guarantees of frequentist FDR control. This would necessitate operating under the most adversarial violation of (\textbf{A1}) --- not only is the null proportion $\pi_0(x)$ not bounded away from $0$ and $1$, but it is precisely equal to either $0$ or $1$ everywhere. Under these conditions PRx consistency can no longer be guaranteed. Despite this, we find in Section \ref{sec:sims} that PRx testing maintains frequentist FDR control while also attaining competitive power, suggesting a high degree of robustness towards mis-specification of (\textbf{A1}). In fact, the power gap between PRx and competing methods noticeably widens in the frequentist setting. This is not surprising; it is a well-documented phenomenon that empirical Bayes multiple testing procedures perform well under both frequentist and Bayesian model specifications \citep{bogdan2008comparison}. 

We now proceed by deriving error rate control properties for our rejection rules. A standard calculation shows that the oracle SC thresholding rule is power-optimal among rejection rules that threshold the posterior FDR. For a review of this see the supplementary. Theorem \ref{thm:consistency}, combined with this result, establishes complementary parts of the oracle picture: The former shows that PRx consistently recovers the oracle local false discovery rate $\ell(z,x)$, while the latter shows that, if this oracle quantity were known, then it could be used to achieve power optimal posterior FDR control. The remaining question is how these properties carry over when the unknown $\ell(z,x)$ is replaced by its PRx estimate.

\begin{theorem}[Asymptotic Bayesian FDR Control]\label{thm:fdrcontrol}
    Define $\hat{\ell}^{(h,j)}(z,x)$ to be the PRx estimate obtained by initializing $\hat{\Psi}_0^{(h)}(\cdot|x) = \hat{\pi}_{0,0}^{(h)}(x)\delta_0 + (1-\hat{\pi}_{0,0}^{(h)}(x))\hat{\psi}_0^{(h)}(\cdot|x)$, with $\hat{\psi}_0^{(h)}(\cdot|x) = 0$ on $[-h,h]$ and $0<\hat{\pi}_{0,0}^{(h)}(x)<1$, for fold $j$. Define $\hat{\ell}^{(h)}(z_i,x_i) = \hat{\ell}^{(h,1)}(z_i,x_i)$ for $i\in I_2$ and $\hat{\ell}^{(h)}(z_i,x_i) = \hat{\ell}^{(h,2)}(z_i,x_i)$ for $i\in I_1$. Also define $\ell(z_i,x_i)$ to be the oracle local false discovery rate evaluated at the data point $(z_i,x_i)$. Let $R_n^{h}$ denote the set of rejected indices under the split-data rejection rule with estimates $\{\hat{\ell}^{(h)}(z_i,x_i)\}_{i=1}^n$. Let us assume (\textbf{A1}), (\textbf{A2}), (\textbf{A3}). Additionally, letting $G(t) = \text{Pr}(\ell(z,x)\leq t)$, let us assume that for some $t'<\alpha$, $G(t') >0$. Define $\mathcal{V}_n^h:=\sum_{i\in R_n^{h}} \mathbbm{1}\{{H_i=0}\}$. Then, for all $h\leq \delta$, the following holds:
    $$\limsup_{n\rightarrow\infty}\mathbb{E}\bigg[\frac{\mathcal{V}_n^h}{|R_n^h|\vee1}\bigg]\leq \alpha$$
\end{theorem}

\begin{proof}
    Fix $h\leq \delta$, and drop the $h$ notation. Let $\mathcal{D}_n$ denote the full dataset $(z_i,x_i)_{i=1}^n$, and let $\mathcal{D}_1$ and $\mathcal{D}_2$ denote the data in fold $1$ and fold $2$, respectively. We can decompose a data-conditional expectation as follows:
    \begin{equation}
        \begin{split}
            \mathbb{E}\bigg[\frac{\mathcal{V}_n}{|R_n|\vee1} |\mathcal{D}_n\bigg] &= \frac{1}{|R_n|\vee 1}\sum_{i\in R_n}\hat{\ell}(z_i,x_i) + \frac{1}{|R_n|\vee 1}\sum_{i\in R_n}(\ell(z_i,x_i) - \hat{\ell}(z_i,x_i))\\
            &\leq \alpha + \frac{1}{|R_n|\vee1}\sum_{i\in R_n}|\ell(z_i,x_i) - \hat{\ell}(z_i,x_i)|
        \end{split}
    \end{equation}
    Define the event $A_n:=\{|R_n|\geq n\epsilon\}$ for some fixed $0<\epsilon<G(t')$. Under the assumption that $G(t')>0$, we have, for some $\epsilon < G(t')$, $\text{Pr}(|R_n| < n\epsilon)\rightarrow 0$ (see the supplementary). We can write the realized Bayesian FDR as the following:
    \begin{equation}
        \begin{split}
            \mathbb{E}\bigg[\frac{\mathcal{V}_n}{|R_n|\vee1}\bigg] = \underbrace{\mathbb{E}\bigg[\mathbb{E}\bigg[\frac{\mathcal{V}_n}{|R_n|\vee1} 1_{A_n}|\mathcal{D}_n\bigg]\bigg]}_{\text{Term (I)}} + \underbrace{\mathbb{E}\bigg[\mathbb{E}\bigg[\frac{\mathcal{V}_n}{|R_n|\vee1}1_{A_n^c }|\mathcal{D}_n\bigg]\bigg]}_{\text{Term (II)}}
        \end{split}
    \end{equation}
    Let us begin with Term (II). Notice that, since $\mathcal{V}_n/R_n\leq 1$, the expectation is bounded from above by $\text{Pr}(A_n^c)\rightarrow 0$. For Term (I) we take the following steps:
    \begin{equation}
        \begin{split}
           \mathbb{E}\bigg[\mathbb{E}\bigg[\frac{\mathcal{V}_n}{|R_n|\vee1} 1_{A_n}|\mathcal{D}_n\bigg]\bigg] &\leq \mathbb{E}\bigg[\bigg(\alpha + \frac{1}{|R_n|\vee1}\sum_{i\in R_n}|\ell(z_i,x_i) - \hat{\ell}(z_i,x_i)|\bigg)1_{A_n}\bigg]\\
           &\leq
           \alpha + \frac{1}{\epsilon n}\sum_{i=1}^n \mathbb{E}|\ell(z_i,x_i) - \hat{\ell}(z_i,x_i)|\\
           & \leq \alpha + \frac{1}{\epsilon}\bigg(\frac{1}{n_1}\sum_{i\in I_1}\mathbb{E}[|\ell(z_i,x_i) - \hat{\ell}^{(2)}(z_i,x_i)|] +\underbrace{(\cdots)}_{I_2\text{ term}}\bigg)
        \end{split}
    \end{equation}
We argue that $\frac{1}{n_1}\sum_{i\in I_1} \mathbb{E}|\ell(z_i,x_i) - \hat{\ell}^{(2)}(z_i,x_i)|\rightarrow 0$ almost surely. Notice that, conditional on $\mathcal{D}_2$, $|\ell(z_i,x_i) - \hat{\ell}^{(2)}(z_i,x_i)|; i\in I_1$ are i.i.d. random variables, and thus we have that $\frac{1}{n_1}\sum_{i\in I_1} \mathbb{E}|\ell(z_i,x_i) - \hat{\ell}^{(2)}(z_i,x_i)| = \mathbb{E}|\ell(z_k,x_k) - \hat{\ell}^{(2)}(z_k,x_k)|$ for any fixed $k\in I_1$. By Theorem \ref{thm:consistency} this approaches zero almost surely. The exact same argument applies to $I_2$, and the proof is complete.
\end{proof}

It is worth briefly discussing why the data-splitting mechanism of two-fold PRx is necessary to attain asymptotic BFDR control. The Bayesian FDR $\mathbb{E}[\mathcal{V}_n^h/\min(R_n^h,1)]$ can be bounded from above by $\alpha + \xi_n$, where the error term $\xi_n$ collapses almost surely to zero if $\hat{\ell}_n(z_i,x_i)\xrightarrow{L^1}\ell(z_i,x_i)$. Theorem $\ref{thm:consistency}$ is limited in this regard; consistency of $\hat{\ell}_n(z,x)$ to $\ell(z,x)$ does not imply consistency of $\hat{\ell}_n(z_i,x_i)$ to $\ell(z_i,x_i)$, since the data pair $(z_i,x_i)$ is used both as an evaluation point and within the PRx recursion internal to $\hat{\ell}_n(\cdot,\cdot)$. Data splitting resolves this issue by explicitly separating $(z_i,x_i)$ from the data used to train the local false discovery rate estimator. 

As seen in Section \ref{sec:sims}, two-fold PRx leads to stricter BFDR control at the cost of slightly lower power, while the SC thresholding rule tends to reject more aggressively at the cost of occasionally exceeding the $\alpha$ control level. If the practitioner wishes to use a method that provably controls asymptotic BFDR, then our suggestion is to use two-fold PRx. If, on the other hand, they wish to avoid the randomness induced by splitting the data, or want to work with only a single estimate of the local false discovery rate, then the SC rejection rule offers a simpler, more straightforward alternative. Regardless of the decision, Section \ref{sec:sims} suggests that both procedures exhibit excellent performance in comparison to already-existing methods. 

PRx derives its error rate control precisely from the consistency of its estimates of $\ell(z,x)$. It does not rely on masking or calibration mechanisms, which often lead to robustness of FDR control guarantees to model mis-specification \citep{chao2021adapt, leung2022zap}. PRx therefore trades model robustness for a greater emphasis on faithfully recovering the data generating process. Despite this, the numerical studies in Section \ref{sec:sims} show that this trade-off can be highly beneficial even in data regimes that lie outside the assumptions of the formal theory.

\section{Discussion}\label{sec:discussion}

The central aim of PRx testing is to recover the entire local false discovery rate surface $\ell(z,x)$. Recovery of this quantity requires execution of two difficult statistical tasks: conditional density estimation, and recovery of a latent null proportion. PRx offers an elegant solution to this; Theorem \ref{thm:consistency} shows that hole-adjusted PRx consistently recovers the oracle quantity $\ell(z,x)$ under a separation condition on the latent mixing distribution. The work of \cite{sun2007oracle} then describes what is gained if this oracle quantity were known: the SC rule achieves power optimality among rules controlling posterior FDR. Theorem \ref{thm:fdrcontrol} then returns to the feasible problem and shows that replacing $\ell(z_i,x_i)$ with its corresponding cross-fitted PRx estimates retains asymptotic BFDR control. 

%Theorem \ref{thm:oracle-pathwise} then describes what is gained if this oracle quantity were known: the SC rule not only controls posterior FDP by construction, but its realized FDP is asymptotically bounded by the nominal level almost surely.

%Thus the theory separates naturally
%into recovery of the oracle score, properties of the oracle decision
%rule, and transfer of error-rate control to an estimated rule.

Such an emphasis on estimating $\ell(z,x)$ distinguishes PRx from several closely related covariate-adaptive procedures. In ZAP and AdaPT-GMM, the fitted model is used to improve the adaptation of the rejection rule, while frequentist FDR validity is protected by a separate calibration or masking argument. PRx instead treats the conditional mixing distribution as the inferential target, and its asymptotic guarantees are built on its consistent recovery. These represent two different routes to covariate-assisted multiple testing: one insulates error-rate validity from potential misspecification of the fitted working model, while the other uses nonparametric estimation within a rich class of densities to recover the underlying conditional mixture and directly leverage the estimated local false discovery rates. Neither approach should be expected to dominate universally. Our numerical results suggest that the latter strategy can yield substantial power gains when the covariates strongly influence the prevalence, direction, or shape of the non-null distribution.

Estimating the full $\ell(z,x)$ surface may also be useful for reasons that are not captured by power alone. When signals occupy coherent regions of the covariate space, a localized testing procedure can produce rejection sets that inherit this geometry, forming connected or structured regions rather than collections of isolated discoveries. Such structure can be scientifically meaningful when the covariates themselves have a spatial or geometric interpretation, as in imaging or other related settings. PRx does not explicitly impose any topology or connectivity on the rejection set; rather, such structure may emerge from learning how the conditional distribution of the test statistic varies over the covariate space. Developing formal guarantees for recovery of these geometric features would be an interesting direction for future work.

Several limitations of the present theory also point to natural extensions. Exact recovery of the point mass at zero by hole-adjusted PRx currently relies on the separation assumption (\textbf{A3}). The conservative construction in Lemma \ref{lem:conservativeconsistency} shows that this assumption can be avoided at the price of a shrinking-neighborhood limit, and our simulations indicate good performance even when substantial slab mass lies near zero. However, exact guarantees for naive PRx without a separation condition remain unattained. Likewise, the fixed-label simulations deliberately move outside the two-groups formulation used in our theory and should be viewed as evidence of empirical robustness rather than as a claim of uniform fixed-configuration FDR control. Finally, extending the theory to cover data-selected PRML tuning, empirical-null estimation, and dependence among the test statistics would bring the formal guarantees closer to real applications.

\bibliographystyle{abbrvnat}
\bibliography{mybib}
\end{document}